\documentclass{amsart}

\usepackage[linesnumbered,lined,ruled]{algorithm2e}
\usepackage{stmaryrd}
\usepackage{listings}

\newtheorem{theorem}{Theorem}
\newtheorem{lemma}[theorem]{Lemma}
\newtheorem{heuristic}[theorem]{Heuristic}

\theoremstyle{definition}
\newtheorem{definition}[theorem]{Definition}

\theoremstyle{remark}
\newtheorem{remark}[theorem]{Remark}

\numberwithin{equation}{section}

\newcommand{\Orth}[1]{{#1^{\bot}}}
\newcommand{\transp}[1]{{#1^{\top}}}
\newcommand{\Norm}[1]{{\|#1\|}}
\newcommand{\AltNorm}[1]{{\llbracket#1\rrbracket}}
\newcommand{\rank}[1]{{\mbox{rk}(#1)}}
\newcommand{\pot}[1]{{\mbox{Pot}(#1)}}

\newcommand{\Scal}[2]{{\langle #1| #2 \rangle}}
\newcommand{\ZZ}{\mathbb{Z}}
\newcommand{\QQ}{\mathbb{Q}}
\newcommand{\RR}{\mathbb{R}}

\newcommand{\bb}{\mathfrak{b}}
\newcommand{\sig}{\sigma}

\begin{document}

\title{Indefiniteness makes lattice reduction easier}


\author{Antoine Joux}
\address{CISPA Helmholtz Center for Information Security, \\ Stuhlsatzenhaus 5, 66123 Saarbr\"ucken, Germany.}
\email{joux@cispa.de}


\date{}


\begin{abstract}
  Since the invention of the famous LLL algorithm, lattice reduction
  has been an extremely useful tool in computational number
  theory. By construction, the LLL algorithm deals with lattices
  living in a vector space endowed with a positive definite scalar
  product. However, it seems quite nature to ask about the indefinite
  case, where the scalar product is replaced by an arbitrary quadratic
  form, possibily indefinite. This question was considered
  independently in two lines of work. One by G{\'a}bor Ivanyos and {\'A}gnes
  Sz{\'a}nt{\'o} and one by Denis Simon. Both lead to an algorithm
  that generalizes LLL and whose performance is
  very similar to LLL, i.e. a polynomial-time algorithm that
  approximates the shortest vector within an approximation factor
  exponential in the dimension. Denis Simon achieves an approximation
  factor close to that of LLL under the assumption that no isotropic vectors arise
  during reduction. G{\'a}bor Ivanyos and {\'A}gnes
  Sz{\'a}nt{\'o} show that it is possible to avoid isotropic
  vectors altogether, at the cost of a somewhat worse approximation factor.

  In this paper, we revisit the reduction of indefinite lattices and
  conclude that it can lead to much better reduced representations
  that previously thought. We also conclude that the approximation
  factor depends on the signature of the indefinite lattice rather
  than on its dimension.
\end{abstract}

\maketitle
\section{Introduction}
Lattice reduction has been a core tool in computer science since the
invention of the Lenstra--Lenstra--Lov{\'a}sz (LLL)
algorithm~\cite{lenstra1982factoring} in 
1982. It has a wide range of applications in number theory and
cryptography. Lattice reduction which aims at finding a rather short and
orthogonal basis of a given lattice can be equivalently reformulated
as reducing a positive definite binary quadratic form in several variables. The
two dimensional or bivariate case has been well studied
since Gauss in~\cite[Sectio Quinta]{Gauss1801}. Later, Korkine and
Zolotareff explored the notion of extremal positive definite
forms~\cite{KorkineZ1877}, getting concrete bounds with up to five
variables. However, the algorithmic question remained open until LLL,
for more details, see~\cite{smeets2009history}.

In this paper, we are focusing on a generalisation of lattice
reduction or quadratic form reduction when then positive definiteness
constraint is lifted. This is a natural question to ask and it is not
the first time it is raised. Even in small dimension, the study of
such indefinite form is a classical topic also started by
Gauss. Markoff~\cite{Markoff1879} continued the work of
Korkine-Zolotareff by studying extremal cases of indefinite
forms. Algorithmically, as far as we know, the first
investigation of the idea of generalizing LLL was by G{\'a}bor Ivanyos and {\'A}gnes
Sz{\'a}nt{\'o}, with an article~ \cite{IvanyosSz96} submitted in~1993 and published
in~1996. Another, seemingly independent, study was submitted by Denis
Simon~\cite{Simon05}  in 2003 and published in 2005. Both articles assume that the
lattice is given by its Gram matrix, i.e, a symmetric matrix formed of
the ``scalar products'' of all pairs of vectors. They further suppose that
this matrix is integral and has full rank. Simon performs a simple
adaptation of LLL, by putting absolute values around ``squared-norms'',
i.e. scalar product of a vector with itself, whenever they are used in
a comparison. This is necessary since negative  ``squared-norms'' may
appear. Somewhat surprinsingly,  it is enough to generalize the results of LLL under
the assumption that no orthogonalized vector of norm $0$ is
encountered over the course of the computation. This is required since, such vectors,
called isotropic vectors, would then cause division by zero in the
Gram-Schmidt orthogonalization routine used by LLL. Simon also
proposes number theoretic applications which are based of his
indefinite algorithm in small dimension.

Additionally, Simon shows in the last paragraph of~\cite[Theorem
1.4]{Simon05} that there is a small improvement for the approximation
factor compared to the positive definite case. This comes form sign
changes in the ``squared-norms'' of 
orthogonal vectors. In Remark~1.5, he indicates that having many such
sign changes would yield further improvement.

In their article, Ivanyos and Sz{\'a}nt{\'o} propose an approach very
similar to Simon's. However, they introduced a few modifications in
the algorithm to completly avoid isotropic vectors. The cost of their
modifications is that some specific cases make the approximation
factor of the algorithm become slightly worse. However, as with LLL,
its remains exponential in the dimension.

In this paper, we propose two main contributions. First, from a theory
point of view, we give a more general definition of indefinite
lattices that does not require integral matrices. Second, in our main,
algorithmic part, which again deals with integral matrices, we
introduce a new algorithm which take advantage of the correspondance
between lattices and quadratic forms to outperform preexisting
algorithms. In particular, this facilitates the creation of sign
changes that Simon remarked upon. With also report on a first
implementation of this algorithm and highlight some of its results.

 \section{Preliminaries}
\subsection{Standard lattices and the notion of reduction} 
A {\bf lattice} is a discrete additive subgroup of $\RR^n$, together
with the norm arising from a scalar product on $\RR^n$. Often, one
considers the standard scalar product, but any bilinear symmetric real
positive-definite quadratic form can be used. Note that the notion of
discreteness used above implicitly refers to the topology induced by
the chosen scalar product. In many applications, we consider
subgroups of $\ZZ^n$, together with a scalar product that takes
integral values on $\ZZ^n$. This is simpler to implement since it 
avoids worrying about loss of precision
issues. It also guarantees the discreteness condition.

To specify a lattice, one needs to give both a $\ZZ$-basis of the
subgroup of $\RR^n$ and a bilinear symmetric matrix $G$ of dimension
$n\times n$ that induces
the scalar product. Namely if $\vec{x}$ and $\vec{y}$ are vectors in
$\RR^n$ then we define:
$$
\Scal{\vec{x}}{\vec{y}}_G=\transp{\vec{y}}\cdot G \cdot \vec{x}.
$$
The standard scalar product is induced by letting $G$ be the identity
matrix. When $G$ is clear by context, we lighten the notation by writing
$\Scal{\vec{x}}{\vec{y}}$ instead of $\Scal{\vec{x}}{\vec{y}}_G.$

It is also possible to accept more generality and start from a
generating family of vectors together with a scalar product matrix
$G$. In that situation, we still want to end up with a basis.

In implementation of LLL, the two most common choices are the
following:
\begin{itemize}
\item Assume the use of the standard scalar product ,i.e., $G$ is
  implicitly set to the identity matrix and provide either a basis
  (as in~\cite[Alg.~$2.6.3$]{Cohen2013} or a generating family  (as
  in~\cite[Alg.~$2.6.8$]{Cohen2013}) for the lattice.
\item Explicitly give a Gram matrix and assume that the lattice is
  spanned by the canonical basis, as in~\cite[Alg.~$2.6.7$]{Cohen2013}.
\end{itemize}

However, it is easy to work with a Gram matrix plus a basis or
generating family both given explicitly. For example, this choice is
made in~\cite{espitau2020certified}.

\subsection{Symmetric bilinear forms on real vector spaces}
A {\bf bilinear form} on a finite dimensional real vector space $V$ is a map $\bb$ from
$V\times V$ to $\RR$ such that, for all $\vec{x_1}$, $\vec{x_2}$, $\vec{y_1}$, $\vec{y_2}$
  in $V$ and $\lambda$ in $\RR$ we have:
\begin{align}
  &\bb(\vec{x_1}+\vec{x_2}, \vec{y_1})  =\bb(\vec{x_1}, \vec{y_1})  +  \bb(\vec{x_2}, \vec{y_1})  \\
  &\bb(\vec{x_1}, \vec{y_1}+\vec{y_2}) =\bb(\vec{x_1}, \vec{y_1})  +  \bb(\vec{x_1}, \vec{y_2}) 
                                           \quad \mbox{and} \\
  &\lambda \bb(\vec{x_1}, \vec{y_1})  = \bb(\lambda \, \vec{x_1}, \vec{y_1}) = \bb(\vec{x_1}, \lambda \,\vec{y_1}).
\end{align}

The form $\bb$ is {\bf symmetric} if for all $\vec{x}$, $\vec{y}$ we
have:
\begin{align}
  &\bb(\vec{x}, \vec{y})  =\bb(\vec{y}, \vec{x}).
\end{align}

By abuse of language, we may refer to $\bb(\vec{x}, \vec{y})$ as the
``scalar product'' of $\vec{x}$ and $\vec{y}$ and to $\bb(\vec{x}, \vec{x})$ as the
``squared-norm'' of $\vec{x}$.\\

Let $\bb$ be a symmetric bilinear form on a real vector space $V$ and
let $W$ be an arbitrary subspace of $V$ we define the ${\bf
  orthogonal}$ subspace $\Orth{W}$ as:
$$
\Orth{W}=\{ \vec{x} \in V \ |\  \forall \vec{y}\in W: \bb(\vec{x},\vec{y})=0\}.
$$

The vector space $\Orth{V}$ is called the {\bf kernel} of $\bb$. If
the kernel of $\bb$ is non-trivial, we say that $\bb$ is
{\bf degenerate}. If $\vec{v_1}, \cdots, \vec{v_d}$ is a basis of $V$,
we can represent  $\bb$ in this basis via the real symmetric matrix:
$$
G_{\bb}=\begin{pmatrix}
  \bb(\vec{v}_1,\vec{v}_1) & \bb(\vec{v}_1,\vec{v}_2) & \cdots
  & \bb(\vec{v}_1,\vec{v}_d)\\
   \bb(\vec{v}_2,\vec{v}_1) & \bb(\vec{v}_2,\vec{v}_2) & \cdots
  & \bb(\vec{v}_2,\vec{v}_d)\\
  \vdots &\vdots & \ddots &\vdots\\
   \bb(\vec{v}_d,\vec{v}_1) & \bb(\vec{v}_d,\vec{v}_2) & \cdots
  & \bb(\vec{v}_d,\vec{v}_d)
\end{pmatrix}.
$$
The form $\bb$ is non-degenerate if and only if $G_{\bb}$ has full
rank.

A vector $\vec{x}$ in $V$ is said to be {\bf isotropic} for $\bb$ when
$\bb(\vec{x},\vec{x})=0.$ The quadratic form $\bb$ is {\bf positive
  definite} if and only the following two conditions apply:
\begin{itemize}
\item For all $\vec{x}$ in $V$, we have $\bb(x,x)\geq 0.$
\item If $\vec{x}$ is isotropic then $\vec{x}=\vec{0}.$
\end {itemize}
Furthermore, if $-\bb$ is positive definitive, we call $\bb$ a {\bf
  negative definite} form.

If $\bb_1$ is a symmetric bilinear form on $V_1$, and $\bb_2$ is a
symmetric bilinear form on $V_2$, we can define a bilinear form on the
direct sum $V_1\oplus V_2$. We denote this form as $\bb=\bb_1\oplus
\bb_2$ and call it the orthogonal sum of $\bb_1$ and $\bb_2$. It is
defined as follows:
$$
\forall (\vec{x}_1, \vec{x}_2, \vec{y}_1, \vec{y}_2)\in V_1\times
V_2\times V_1\times V_2: \bb(\vec{x}_1+\vec{x}_2,
\vec{y}_1+\vec{y}_2)= \bb_1(\vec{x}_1, \vec{y}_1) +\bb_2(\vec{x}_2, \vec{y}_2).
$$

If $W$ is a vector subspace of $V$, we denote by $\bb_{|W}$ the
restriction of $\bb$ to $W$, i.e. the bilinear from $W\times W$ that
sends $\vec{x}$ and $\vec{y}$ in $W$ to $\bb(\vec{x}, \vec{y})$.

Recall that by the spectral theorem, the matrix $G_{\bb}$ of a form
$\bb$ on $V$  can be diagonalized. Let $V_0$ be the kernel of $V$ and
$V_{+}$ and $V_{-}$ be the vector spaces respectively spanned by the
eigenvectors corresponding the all positive eigenvalues,  resp.  all
negative eigenvalues. As a direct consequence, we have:
\begin{theorem}
  If $V_0$, $V_{+}$ and $V_{-}$ are as above, then:
  $$\bb=\bb_{|V_0}\oplus \bb_{|V_+}\oplus \bb_{|V_-}.$$

Furthermore, $\bb_{|V_0}$ is the constant zero form,  $\bb_{|V_+}$ is
positive definite and $\bb_{|V_-}$ is
negative definite.
\end{theorem}
As a direct corollary, any non-degenerate form is the orthogonal sum of
a positive definite and a negative definite form.

Using the above theorem, we define the signature of the form $\bb$ as
$\sig(\bb)=|\dim{V_+}-\dim{V_-}|.$ This is the absolute value of
the difference between the number of positive and negative
eigenvalues. We put an absolute value here to emphasize the fact that $\bb$
and $-\bb$ induce the same geometry\footnote{Likewise, we also have scaling
  invariance, i.e for any
  non-zero real $\lambda$, the form $\lambda\,\bb$ induces the
  same geometry.} on $V$.

\subsection{Indefinite lattices} Let $V$ be a real vector space
equipped with a symmetric bilinear form $\bb$. Let $\vec{v}_1$,
\dots, $\vec{v}_d$ be vectors in $V$. We consider the additive
subgroup:
$$L=\ZZ\,\vec{v}_1+\cdots+\ZZ\,\vec{v}_d,$$
and want to define a notion of discreteness.

Let us first consider an unsatisfactory approach that motivate the
choice we make.
This idea would be to request that  the ``squared-norm'' of
any non-isotropic vector should be bounded 
away from $0$. We can see that this would, for example, fail for
the following acceptable looking Gram matrix:
$$
\begin{pmatrix}
  1 & 0 \\
  0 & -e^2
\end{pmatrix}.
$$
As the direct sum of two orthogonal definite lattices, one positive
and one negative, this is a very reasonable candidate for an
indefinite lattice. However, calling the two vectors $\vec{u}$ and
$\vec{v}$ we see that:
$$
\bb(a\vec{u}+b\vec{v}, a\vec{u}+b\vec{v})=a^2-(be)^2=(a+be)(a-be).
$$
Thanks to~\cite{davis1978rational}, we know that there are
infinitely many pairs of positive integers $(a,b)$ such that:
$$
\left|e-b/a\right|<\frac{\log{\log{a}}}{a^2\log{a}}.
$$
As a consequence, there are infinitely many pairs with:
$$
|\bb(a\vec{u}+b\vec{v}, a\vec{u}+b\vec{v})| <\frac{4\,\log{\log{a}}}{\log{a}}.
$$
So the (non-zero) ``squared-norm'' of $a\vec{u}+b\vec{v}$ can be
made arbitrarily small.
\\

The same counter-example, would be to rule out bounding all
values of the bilinear form away from $0$. However, if we fix a basis
of the vector space $(\vec{w}_1,\cdots,\vec{w}_d)$ and consider the
quantity:
$$M_w(\vec{v})=\max{\left(|\bb(\vec{v},\vec{w}_1)|,\cdots,
    |\bb(\vec{v},\vec{w}_d)\right|)},$$
we can insist that $M_w(\vec{v})$ should be either $0$ or bounded away
from $0$. Unfortunately, this definition is not a coordinate-free.
However, we could not find a better alternative.

We can formulate the idea slightly differently by defining:
\begin{equation}\label{eqn:altnorm}
  \AltNorm{\vec{x}}_{\bb}=\Norm{G_\bb\cdot \vec{x}}.
\end{equation}

It is clear that this is a norm, which  induces a  topology, on the
direct sum $V_{+}\oplus V_{-},$ even though it is identically $0$ on
$V_{0}$. We easily see that:
\begin{lemma} \label{th:degen}
  For any $\vec{x}\in V$:
  $$
  \AltNorm{\vec{x}}_{\bb}=0 \quad \Leftrightarrow \quad \vec{x}\in V_0.
  $$
\end{lemma}

Using this, we call $L$ an {\bf indefinite lattice} if and only if
non-zero values of $\AltNorm{\vec{x}}_{\bb}$ are bounded away from $0$.
Equivalently, there exists a real $\epsilon_L>0$
such that for all $\vec{x}$ in $L$, the condition
$\AltNorm{\vec{x}}_{\bb}<\epsilon_L$ implies that
$\AltNorm{\vec{x}}_{\bb}=0$.
In particular, thanks to Lemma~\ref{th:degen}, if $\AltNorm{\vec{x}}_{\bb}<\epsilon_L$ then
$\vec{x}\in V_0.$ We refer to the existence of $\epsilon_L$ as begin the
{\bf discreteness} condition and call the minimal possible
$\epsilon_L$ the {\bf discreteness parameter} of the (indefinite)
lattice. Note that $\epsilon_L$ depends on the specific choice of
representation of $L$ by its Gram matrix. It changes if we replace
$G$ by $\transp{U}\cdot G\cdot U$ (with $U$ unitary). However, $L$
remains discrete under such a change of basis.\\

\begin{remark}
  It can be noted that this notion of discreteness is equivalent to
  requesting that the ordinary subgroup associated to the (possibly
  degenerate) Gram matrix $G^2$, which is positive definite on
  $V_{+}\oplus V_{-}$ should be a lattice. Again this is not
  coordinate-free.
\end{remark}

Let $V_L$ denote the vector space spanned by $L$. We can restrict
ourselves to consider the form $\bb_{|V_L}.$ For ease of notation, we
now assume that $V=V_L$ and write the form as $\bb_{L}.$ To this form,
we associate, as before, the real symmetric Gram matrix:
$$
G_{L}=\begin{pmatrix}
  \bb(\vec{v}_1,\vec{v}_1) & \bb(\vec{v}_1,\vec{v}_2) & \cdots
  & \bb(\vec{v}_1,\vec{v}_d)\\
   \bb(\vec{v}_2,\vec{v}_1) & \bb(\vec{v}_2,\vec{v}_2) & \cdots
  & \bb(\vec{v}_2,\vec{v}_d)\\
  \vdots &\vdots & \ddots &\vdots\\
   \bb(\vec{v}_d,\vec{v}_1) & \bb(\vec{v}_d,\vec{v}_2) & \cdots
  & \bb(\vec{v}_d,\vec{v}_d)
\end{pmatrix}.
$$

If $G_L$ is an integral matrix, we say that $L$ is an {\bf integral
  indefinite lattice.} Note that if $\bb_{L}$ is either positive and
negative definite, we recover the definition of standard lattices and
standard integral lattices.

When $\bb_L$ is non-degenerate, we define its {\bf Gram-determinant} as
the (non-zero) determinant of the matrix $G_{L}.$ For a standard lattice,
this is the square of the usual determinant. Note that for an indefinite
lattice, the Gram-determinant can be negative.

Our definition generalizes of the notion of indefinite
lattices that was considered by  by Ivanyos and
Sz{\'a}nt{\'o} in~\cite{IvanyosSz96} and by  Simon in~\cite{Simon05}. Indeed, they only
considered the case of non-degenerate integral 
symmetric matrices, which are clearly discrete according to our definition.

\subsubsection{Extension to Hermitian lattices}
After considering lattices associated with arbitrary real symmetric
matrices, the next logical step would be to consider an Hermitian
scalar product, i.e. a matrix with complex entries whose transpose is
equal to its conjuguate. In that situation, its would be natural to
consider lattices formed of linear combinations with coefficients in
$\ZZ[i]$. However, it is easy with a doubling of the dimension to
transform this back into an indefinite lattice. Indeed, it suffices to
encode the complex $a+ib$ into the $2\times 2$ matrix:
$$
\begin{pmatrix}
  a& -b\\
  b&a
\end{pmatrix}.
$$
We can see that replacing each entry of an Hermitian matrix by such a
small block yields a real symmetric matrix.

\subsection{Binary quadratic forms} \label{sec:binforms}
The study of binary quadratic forms and their reduction is a classical
topic in number theory which has attracted a lot of attention. It is
deeply related to lattice and indefinite lattice reduction in
dimension~2.
\begin{definition}
  Let us start with a few standard definitions.

  \begin{itemize}
  \item An {\bf integral binary quadratic form} is homogeneous bivariate
    polynomial of degree two: $Q(x,y)=a\,x^2+b\,xy+c\,y^2,$ with $a$,
    $b$ and $c$ integers. It is often denoted by the shortened notation $(a,b,c)$.
  \item The {\bf discriminant} of $(a,b,c)$ is $\Delta=b^2-4ac.$
  \item We say that the form $Q(x,y)$ represents a value $m$  if there
    exists integers $x_0$ and $y_0$ such that $Q(x_0,y_0)=m.$
  \item The representation of $v$ by $x_0$ and $y_0$ is called {\bf
      primitive} if $x_0$ and $y_0$ are coprime.
  \item Given four integers $\alpha, \beta, \gamma, \mbox{and}\
    \delta$ such that $\alpha\delta-\beta\gamma\neq 0$ we can define a
    new quadratic form by a simple change of variables:
    $$ Q'(x,y)=Q(\alpha x+\beta y, \gamma x +\delta y).$$
    The coefficients $(a',b',c')$ of $Q'$ can easily be computed and
    the transformed discriminant is
    $\Delta'=(\alpha\delta-\beta\gamma)^2\Delta.$
  \item If $(\alpha\delta-\beta\gamma)=\pm 1$ we say that $Q$ and
    $Q'$ are {\bf equivalent}. It is easy to check that this is indeed an
    equivalence relation on quadratic forms.
  \item When   $(\alpha\delta-\beta\gamma)=1$, we say that $Q$ and
    $Q'$ are {\bf properly equivalent}, otherwise we call the
    equivalence {\bf improper}.
  \end{itemize}
\end{definition}

The reduction of quadratic forms is the process of transforming a
quadratic form into an equivalent form that satisfies additional
properties making it preferable. Reduction is often performed
using only proper equivalence.

In this paper, we need to deviate from standard practice and consider
{\bf rational binary quadratic form} where the coefficients $(a,b,c)$
are no longer limited to integers but can also be rational numbers.
However, this is not a issue since the definition of reduction for
quadratic forms is scale invariant, namely if $(a,b,c)$ and $(a',b',c')$ are
(properly) equivalent then so are $(\lambda a, \lambda b, \lambda c)$
and $(\lambda a', \lambda b',\lambda c').$ Furthermore, the notions of
reduction are such that if $(a,b,c)$
is reduced, then so is $(\lambda a, \lambda b, \lambda c)$ for any $\lambda>0$.
Thus, we can easily move rational form equivalence problem to an
integral instance by multiplying by a positive integer $\lambda$ that clears all
denominators of $(a, b,c)$ and dividing the result $(a',b',c')$ by
$\lambda$ to restore the denominators. For practical purposes, this
is not even needed since all the reduction algorithms we consider are also
scale-invariant and  directly work
with rational forms.

Before going into the details of reduction of quadratic forms, let us
first explain the useful parallel between binary quadratic forms and
two-dimensional lattices. Let: 
$$G=
  \begin{pmatrix}
    a & b/2\\
    b/2&c
    \end{pmatrix},
$$
be the (rational) Gram matrix of lattice given by two vectors
$\vec{u}$ and $\vec{v}$.

Remark that:
$$
\begin{pmatrix}x&y\end{pmatrix}\cdot G \cdot
\begin{pmatrix}x\\y\end{pmatrix}=ax^2+bxy+cy^2.
$$
We say that $G$ {\bf induces} the form $(a,b,c)$.
We easily that the discriminant of the form $(a,b,c)$ is
$\Delta=-4\det{G}.$
Furthermore, let $T$ denote the matrix:
$$
\begin{pmatrix}
  \alpha & \beta\\
  \gamma&\delta
\end{pmatrix},
$$
and let $G'=\transp{T}\cdot G\cdot T$. It is easy to see that $G'$ induces
the form $(a',b',c')$ obtained form $(a,b,c)$ by the change of variables
provided by $\alpha, \beta, \gamma,$ and $\delta$. Furthemore, we have
equivalent induced forms when $\det(T)=\pm 1$ and proper equivalence
when $\det(T)=1.$ This shows that the reduction of binary quadratic
forms and of Gram matrix are just two sides of the same coin. This is
our motivation to use the well studied reduction of indefinite form to
deal with the case where $\det(G)<0$.

To push the analogy even further, let us remark that $(x_0,y_0)$
represents $m$, if and only if, the (indefinite) norm of
$x_0\vec{u}+y_0 \vec{v}$ is $m$. The representation is primitive when
the integer point $(x_0,y_0)$ is visible from the origin in the
lattice $\ZZ^2$.

\subsubsection{Definite forms}
For a quadratic form $(a,b,c)$, remark that it is
definite\footnote{Positive definite if $a>0$ and negative definite otherwise.}
when $\Delta<0$. In that situation, the form has no real
(thus no rational roots).

Classically, e.g. see~\cite[Def $5.3.2$]{Cohen2013}, a positive
definite form is said to be {\bf reduced} when $|b|\leq a 
\leq c.$ Since we also consider negative definite forms, we instead
use the condition $|b|\leq |a| 
\leq |c|.$ There is a simple algorithm that reduces an initial form by
iteratively making improving it. Assuming that the initial
form\footnote{If it is not the case, we replace $(a,b,c)$ by the
  properly equivalent form $(c,-b,a)$. The equivalence can be checked
  by using the unimodular matrix $ \begin{pmatrix}
  0 & 1\\
  -1& 0
\end{pmatrix}.
$}
satisfies $|a|\leq |c|$, we repeatly apply the transformation:
$$
(a,b,c) \longrightarrow (a',b',c')=(a,b+2\lambda a, c+\lambda b+\lambda^2 a),
$$
where $\lambda$ is chosen as the closest integer to $-b/2a$.
This corresponds to applying the unimodular matrix:
$$
\begin{pmatrix}
  1 & \lambda\\
  0& 1
\end{pmatrix}.
$$
If $|a'|>|c'|$, then we turn the form into $(c',-b',a')$ and repeat.
Otherwise, the form $(a',b',c')$ is reduced and we stop.

\subsubsection{Indefinite forms}
For indefinite forms, it is usual to consider only forms where
$\Delta$ is not a square. In particular, this implies $\Delta\neq 0$,
thus $\Delta>0$. This is usually done for integral forms. However, it
is worth noting that the condition $\Delta$ is not a square can be
interpreted as {\it $\Delta$ is not the square of a rational} for forms with
rational coefficients.

In this situation, the polynomial $ax^2+bx+c$ has two real roots but
no rational ones, and thus defines a proper quadratic extension of
$\QQ,$ namely $\QQ(\sqrt{\Delta}).$ This is the rationale behind
the squarefreeness  restriction. Such a squarefree indefinite form is
called {\bf reduced} whenever the following condition, see~\cite[Def
$5.6.2$]{Cohen2013}, applies:
\begin{align}\label{eqn:redcond}
   &\left| \sqrt{\Delta} -2|a|\right| < b <\sqrt{\Delta}.
\end{align}
Since $a$ and $b$ are integers (or rationals for us) and
$\sqrt{\Delta}$ is irrational, it is clear that the two 
inequalities are strict. Furthermore, see~\cite[Prop $5.6.3$]{Cohen2013}, replacing $a$ by $c$ in
Equation~\eqref{eqn:redcond} yields an equivalent condition for reducedness.

Starting from a non-reduced form, we can iteratively apply the following
transformation:
$$
(a,b,c) \rightarrow (c, -b+2c\delta, a-b\delta+c\delta^2),
$$
where $\delta$ is the unique integer that satisfies:
\begin{align*}
  -|c|<-b+2c\delta\leq |c| &\quad \mbox{when}\ |c|> \sqrt{\Delta} \quad
                                         \mbox {and}\\ 
  \sqrt{\Delta}-2|c|<-b+2c\delta<\sqrt{\Delta}&\quad \mbox{otherwise}.
\end{align*}

As long as $|a-b\delta+c\delta^2|<|c|$, we repeat the
transformation. If $|a-b\delta+c\delta^2|\geq |c|,$ we obtain a
reduced form. By contrast with the definite case, we can continue
applying the reduction transformation after reaching our first reduced
form. This creates a cycle of reduced forms and the signs of the first
component of the forms alternate throughout the cycle.

\begin{remark}
  The choice of $\delta$ that differs depending on the relative size
  of $|c|$ and $\sqrt{\Delta}$ is necessary to guarantee that we reach
  a reduced form in polynomial time. A proof is of the complexity is
  given in~\cite[Prop $5.6.6$]{Cohen2013}. For a problematic example
  without this, the reader can, for example, consider the form:
  $$
  (5133516356526721720, -2*5133515988396719824,  5133515620266744327)
  $$
\end{remark}

\subsubsection{\label{sec:missing}The missing cases}
Unfortunately, while the problem was initially considered by
Gauss~\cite[Problema 206]{Gauss1801}, modern literature does not tell
us what to do when~ $\Delta$ is a square. It turns out that we need to distinguish two
cases.

First, consider $\Delta=0$. This can happen in the fully degenerate
situation  $a=b=c=0$, where nothing can or needs to be done. It can
also happens when $b^2=4ac$. In that case, we can apply the definite
case algorithm until we reach a form $(0,0,c)$ where we stop. In
terms of two-dimensional lattices, we have two linearly dependent vectors
and we essentially compute their GCD. This results in a pair formed a
the zero vector and a vector which by itself generates the same
lattice as the two input vectors together. Thus, this vector can be written
as either of the (non zero) input vectors divided by an integer.

Second, consider the case where $\Delta>0$ is the square of a
rational. Here, we would like to adapt the indefinite case. However,
when we look at the conditions given by Equation~\eqref{eqn:redcond},
we see that since $\sqrt{\Delta}$ is now
rational, the inequalities may become large. Furthermore, we need to
choose which of the various options suit us.

The first equality case in condition~\eqref{eqn:redcond} occurs when
$b=0$. Then, the only possible generalisation of 
Equation~\eqref{eqn:redcond} implies:
$$2|a|=\sqrt{\Delta}, \quad \mbox{and}\ c=-a.$$
In that situation, we can create a cycle between the two reduced forms
$(\sqrt{\Delta},0,-\sqrt{\Delta})$ and
$(-\sqrt{\Delta},0,\sqrt{\Delta})$ using the proper equivalence given
by $\begin{pmatrix}
  0 & 1\\
  -1& 0
\end{pmatrix}.$ This creates a sign alternance which is
welcome in the algorithm of Section~\ref{sec:ouralg}.
\\

The second equality case in condition~\eqref{eqn:redcond} is
$b=\sqrt{\Delta}$. We then need $ac=0$. With our application in mind,
we choose to enforce $c=0.$ Applying the proper equivalence given by
 $\begin{pmatrix}
  1 & 0\\
  \lambda& 1
\end{pmatrix},$ the form $(a,b,c)$ is transformed into
$(a+\lambda\,b,b,c).$ For a reduced form, we insist that $2|a|< b$.
Note that if $|a|=b/2$, the form $(\pm b/2,b,0)$ is equivalent to 
$(\pm b/2,0,\mp b/2)$ using the unimodular transform $\begin{pmatrix}
  1 & \mp 1\\
  0& 1
\end{pmatrix}.$
We thus go back to the first case of equality. When $a=0$, we are in the the hyperbolic plane described in
Section~\ref{sec:Hyp-planes}. 
\\

Outside of these two special cases, we also need to slightly modify
the reduction process by
requiring that in the reduction step $$
(a,b,c) \rightarrow (c, -b+2c\delta, a-b\delta+c\delta^2),
$$
the value $\delta$ is chosen as the unique integer that satisfies:
\begin{align*}
  -|c|<-b+2c\delta\leq|c| &\quad \mbox{when}\ |c|> \sqrt{\Delta} \quad
                                         \mbox {and}\\ 
  \sqrt{\Delta}-2|c|<-b+2c\delta\leq\sqrt{\Delta}&\quad \mbox{otherwise}.
\end{align*}

\section{Theoretical results about indefinite lattices}

Our goal for indefinite lattice reduction is to provide a unimodular
transformation that achieves two distinct purposes:
\begin{itemize}
\item Decompose $\bb_L$ into an orthogonal sum of the zero form and
  a non-degenerate form.
\item Put the non-degenerate form in a good basis with stronger
  guarantees than provided by~\cite{IvanyosSz96} or~\cite{Simon05}. 
\end{itemize}
For convenience, our algorithmic results, provided in Section~\ref{sec:ouralg}, only deals
with integral indefinite lattices.

In the present section, we consider the first purpose in the general
case and show how discreteness guarantees a
decomposition via a unimodular into an orthogonal sum of the zero form
and a non-degenerate form. More precisely, we prove the following.

\begin{theorem}\label{thm:zerofactor}
  Let $L$ be a lattice, given by a bilinear form $\bb$ on $V$, a
  family of vectors $(\vec{v}_1, \cdots,\vec{v}_k),$ with discreteness
  parameter $\epsilon_L>0$. We denote by $G_L$ the ($k$-dimensional) Gram matrix of $L$.

  If $\rank{G_L}=\ell$, there exists a unimodular matrix $U$ of dimension
  $k$, such that:
  $$
  \transp{U}\,G_L\,U=
  \begin{pmatrix}
    G_\ell & 0 \\
    0 &0
  \end{pmatrix},
  $$
  where $G_\ell$ is a non-degenerate $\ell$-dimensional Gram matrix.
\end {theorem}
\begin{proof}
  Let $V_0$ denote the right-kernel of $G_L$, it is a real vector space
  of dimension $k-\ell$. Consider the set:
  $$L_0=V_0 \cap \ZZ^n.$$
  $L_0$ is a primitive sublattice of $L$, see~\cite[Lemma
  4]{nguyen2009hermite}.
  Thus, any basis of $L_0$ can be completed into a basis
  of $\ZZ^k$, yielding a unimodular matrix $U$.

  To conclude the proof, we need to check that the dimension of $L_0$
  is equal to $k-\ell$. This requires the discreteness condition
  specified by the parameter $\epsilon_L$.

  We denote by $C$ the largest absolute value of eigenvalues of $G_L$.
  With this notation, for any vector $\vec{x}$ in $\RR^k$, we
  have:
  $$
  \Norm{G_L\cdot \vec{x}}\leq C\cdot \Norm{\vec{x}}.
  $$

  For a proof by contradiction, let's assume that the dimension of $L_0$ is
  smaller than $k-\ell$. Then there exists a vector $\vec{x}$ in $V_0$
  and not in the span of $L_0$. Let $Q$ be a large integer to be
  specified later on. For any positive integer $q$, we denote by
  $\vec{x}_q$ the vector formed by rounding each coordinate of the vector
  $q\,\vec{x}$ to the nearest integer. We know by applying Dirichlet's
  theorem, see~\cite[Theorem 4]{hanrot2009lll}, and summing on the $k$ coordinates that
  there 
  exists a positive integer
  $q\leq Q^k$ such that:
  $$
  \Norm{q\,\vec{x}-\vec{x}_q}_1\leq k/Q.
  $$
  For the same $q$, we thus have:
    $$
  \Norm{q\,\vec{x}-\vec{x}_q}_2^2\leq k/Q^2.
  $$
  And,
  $$
  \Norm{G_L\cdot \left(q\,\vec{x}-\vec{x}_q\right)}\leq C\sqrt{k}/Q.
  $$
  The right-hand side can be made smaller than $\epsilon_L$ by choosing a large
  enough $Q$. 
  Furthermore, since $\vec{x}$ is in $V_0$, we have:
  $$
  \Norm{G_L\cdot \left(q\,\vec{x}-\vec{x}_q\right)}=\Norm{G_L\cdot
    \vec{x}_q}< \epsilon_L
  $$
  Now, since $\vec{x}_q$ is integral, the
  discreteness condition implies that:
   $$
  \Norm{G_L\cdot  \vec{x}_q}=0.
  $$
  Thus $\vec{x}_q$ is in $V_0$ and as an integer vector also in $L_0$.
  \\

  On the other hand, let $d>0$ be the Euclidean distance of $\vec{x}$ to
  the vector space spanned by $L_0$. After scaling, the distance of $q \vec{x}$ to
  this vector space is also larger than $d$. However, we know that:
  $$
  \Norm{q\,\vec{x}-\vec{x}_q}_2^2\leq k/Q^2.
  $$
  Since $\vec{x}_q$ is in the vector space spanned by $L_0$, with a
  large enough $Q$, this yields the desired contraction and finishes
  the proof.
\end{proof}
\begin{remark}
  The unimodular matrix in the above theorem is not unique. However,
  they are related. To see that, assume that we have transformed $G_L$
  into both
  $$
    G_1=\begin{pmatrix}
    G_\ell & 0 \\
    0 &0
  \end{pmatrix}
  \quad \mbox{and} \quad
    G_2=\begin{pmatrix}
    G'_\ell & 0 \\
    0 &0
  \end{pmatrix}.
  $$
  Of course, these matrices follow the relation
  $\transp{U}\,G_1\,U=G_2$ for some unimodular matrix $U$.
  Note that the lower right quadrants of both $G_1$ and $G_2$
  arise from the choice of a basis of the lattice $L_0$, which are
  related by a $k-\ell$ dimensional unimodular matrix $W$. As a
  consequence, we can write:
  $$
  U=\begin{pmatrix}
    U_\ell & 0 \\
    V & W
  \end{pmatrix}.
  $$
  The submatrix $V$ is arbitrary and since $\det{U}=\det{U_\ell}\,
  \det{W}$ the fact that both $U$ and $W$ are unimodular implies that
  $U_\ell$ is unimodular.

  As a consequence, we have $G'_\ell= \transp{U_\ell}\,G_\ell \,
  U_\ell,$ we means that $G_\ell$ and $G'_\ell$ are Gram
  matrices of the same non-degenerate indefinite lattice (in
  different bases).

  We call this lattice the {\bf non-degenerate lattice induced} by
  $G_L$. We can use the induced non-degenerate lattice to define
  several important invariants of indefinite lattice. The first is the
  {\bf (non-degenerate) dimension} of the lattice that we identify
  with the dimension of the square matrix $G_\ell$. The second
  invariant of $L$ is the Gram determinant of the induced lattice. We call it the {\bf
    non-degenerate Gram determinant} of $L$ and denote it by
  $\det_{\neq 0} L$ or $\det_{\neq 0}G_ L.$ The third invariant is the
  signature of $G_\ell$, that we call the {\bf (non-degenerate)
    signature} of $L$.

  While it is clear that dimension and determinant are preserved when
  applying a unimodular transformation, the fact that the signature is
  also an invariant requires more analysis. We provide the proof of
  this fact in Appendix~\ref{th:unimodEigen}.
  
\end{remark}

\section{A bird's eye presentation of LLL and necessary adaptations}
In the present Section, we give a high-level description of the LLL
algorithm that we want to use as a blueprint for our new algorithm. This
description removes low-level details, as they are often bound to
the special case of ordinary lattices.

Given an ordinary lattice, described by a positive definitive
symmetric matrix $G$ and a family of vectors $L$, we want to find a
unimodular matrix $U$ such that:
$$ \transp{U}\cdot \transp{L}\cdot  G\cdot L\cdot  U=\begin{pmatrix}
  G_{\mbox{lll}} &0 \\
  0& 0
\end{pmatrix},
$$
where $G_{\mbox{lll}}$ is a positive definitive symmetric matrix that
is LLL-reduced. The right vectors corresponding the lower-right block
of $0$ remove any linear dependency present in $L$. As a consequence,
we get a basis of
the lattice when the input is a generating family.\\

\begin{algorithm}[H]
\DontPrintSemicolon
\KwData{Description of input lattice }
\KwResult{Unimodular $U$ that reduces the input lattice}
\tcp{All operations in bold implicitly updates $U$ and $L_c$}
Set {\it Current position} $k$ to $1$\;
Copy input $L$ to current lattice $L_c$\;
\While{$k$ not at end of (the non-zero part of) $L_c$}{
  {\bf Size reduce} vector in position $k+1$ \tcp*[f]{by earlier vectors}\;
  {\bf Move away} any zero vector to the end of $L_c$\;
  \tcp{If this happens, restart loop with updated vector at $k+1$}
  {\bf Compute} Orthogonal projections of vectors (using GSO)\;
  {\bf Reduce} the projected $2\times 2$ block of the vectors in
  positions $k$ and $k+1$\;
  \leIf{Reduction is active}{Set $k$ to $\max(k-1,1)$}{Set $k$ to $k+1$}
}
\caption{\label{alg:LLL} High-level description of LLL}
\end{algorithm}
\vspace{1cm}

\begin{theorem}[Proposition 1.6 of~\cite{lenstra1982factoring}] \label{th:LLL}
  For any lattice $L$ of dimension $d$, Algorithm~\ref{alg:LLL}
  outputs, in polynomial time, a basis $(\vec{v}_1, \cdots,\vec{v}_d),$ such that:
  $$
  \Norm{\vec{v}_1}\leq C_{LLL}^{d-1}\,\det(L)^{1/d},
  $$
  where $C_{LLL}=\sqrt[4]{4/3+\epsilon}$ and $\epsilon>0$ can be
  made arbitrally close to zero (by setting some parameter in the algorithm).
\end{theorem}
\begin{proof}
  For completeness, we recall the classical proof of the approximation quality
  from~\cite{lenstra1982factoring}.
  In each projected block, we have:
  $$
  \Norm{\vec{v}_{i}^{*}}^2\leq (4/3+\epsilon)\,\Norm{\vec{v}_{i+1}^{*}}^2.
  $$
  Thus, for all $j$:
   $$
  \Norm{\vec{v}_{1}}^2=\Norm{\vec{v}_{1}^{*}}^2\leq (4/3+\epsilon)^{j-1}\,\Norm{\vec{v}_{j}^{*}}^2.
  $$
  Multiplying all the inequalities, we find:
     $$
  \Norm{\vec{v}_{1}}^{2d}\leq (4/3+\epsilon)^{\sum_{j=1}^{d}(j-1)}\,\prod_{j=1}^{d}\Norm{\vec{v}_{j}^{*}}^2=(4/3+\epsilon)^{d(d-1)/2}\, \det(L)^2,
  $$
and the conclusion follows by taking the $2d$-th root.
\end{proof}

In Algorithm~\ref{alg:LLL}, we do not specify how the projected
$2\times 2$ blocks are computed, so let us recall that this is done
using Gram-Schmidt Orthogonalization.  With this in mind, we see that
in order to adapt LLL to the case of indefinite lattices, we need to
revisit three main ingredients:
\begin{itemize}
\item The Gram-Schmidt Orthogonalization.
\item The reduction of $2\times 2$ (projected) blocks.
\item And, the size reduction.
\end{itemize}
The rest of the present Section discusses these ingredients and how
they can be adapted to our purpose. It also discuss the notion of
hyperbolic planes and how they can fit into reduced bases via
admissible sub-lattices.

\subsection{Generalized Gram-Schmidt Orthogonalization}\label{sec:extendedGSO}
Gram-Schmidt Orthogonalization (GSO) is a fundamental tool for the LLL
algorithm. However, the existence of isotropic vectors can lead to
divisions by zero in GSO, so we need to revisit and generalize the GSO
algorithm to avoid this issue.

For the sake of generality, we first follow Simon~\cite{Simon05} and
adapt the usual Gram-Schmidt 
Orthogonalization algorithm for an indefinite lattice $L$ given by a
symmetric bilinear form $\bb$ and a family of vectors $(\vec{v}_1,
\cdots, \vec{v}_k),$ under the assumption that division by zero does not
occur. The algorithm iteratively construct a family of vectors $(\vec{v}^{*}_1,
\cdots, \vec{v}^{*}_k)$ using the formula:
\begin{equation}\label{eqn:GSObase}
  \vec{v}^{*}_i=\vec{v}_i-\sum_{j=1}^{i-1}\frac{\bb(\vec{v}_i,\vec{v}^{*}_j)}{\bb(\vec{v}^{*}_j,
    \vec{v}^{*}_j)}\, \vec{v}^{*}_j,
\end{equation}
with the empty sum interpreted as $0$ when $i=1$.

The assumption that division by zero do not occur is equivalent to $\bb(\vec{v}^{*}_i,
\vec{v}^{*}_i)\neq 0$ for all $i$.
After the GSO procedure, the following properties are satisfied:
\begin{align}
  \forall i:\quad & \bb(\vec{v}^{*}_i, \vec{v}^{*}_i) =\bb(\vec{v}^{*}_i, \vec{v}_i) \\
  \label{eqn:GSOisortho} \forall i> j: \quad & \bb(\vec{v}^{*}_i, \vec{v}^{*}_j) =0
\end{align}

In our lattice reduction from Section~\ref{sec:ouralg}, we largely but not
completly avoid vectors with $\bb(\vec{v}^{*}_i,
\vec{v}^{*}_i)=0.$ More precisely, whenever processing $\vec{v}_i$ to
compute $\vec{v}^{*}_i,$ we are given a set of bad indices $B_i\subset
[1\cdots i-2],$ such that:
\begin{align}
  \forall j\not\in B_i\ \mbox{and}\  (j-1)\not\in B_i:\quad &
                                                              \bb(\vec{v}^{*}_j, \vec{v}^{*}_j)\neq 0 \\
  \nonumber\\
  \forall j\in B_i: \quad & \bb(\vec{v}^{*}_j, \vec{v}^{*}_j) =0\\
  \nonumber    &\bb(\vec{v}^{*}_{j+1}, \vec{v}^{*}_{j+1}) =0\\
  \nonumber     &\bb(\vec{v}^{*}_{j}, \vec{v}^{*}_{j+1})   \neq 0
                  \quad\quad\quad\quad\mbox{and}\\
   \nonumber\\
 \forall \ 1\leq j<\ell<i: \quad  & \bb(\vec{v}^{*}_j, \vec{v}^{*}_\ell)=
                                0\quad \mbox{unless}\ j\in B_i\
                                \mbox{and}\ \ell=j+1.
\end{align}

We define the set of good indices as $G_i=[1\cdots i]\setminus \{j\,|
j\in B_i\ \mbox{or}\ (j-1)\in B_i\}.$
This lets us replace Equation~\eqref{eqn:GSObase} by:
\begin{equation}\label{eqn:GSOnew}
  \vec{v}^{*}_i=\vec{v}_i-\sum_{j\in G_i}\frac{\bb(\vec{v}_i,\vec{v}^{*}_j)}{\bb(\vec{v}^{*}_j,
    \vec{v}^{*}_j)}\, \vec{v}^{*}_j - \sum_{j\in  B_i}\frac{\bb(\vec{v}_i,\vec{v}^{*}_j)}{\bb(\vec{v}^{*}_j,
    \vec{v}^{*}_{j+1})}\, \vec{v}^{*}_{j+1} - \sum_{j\in  B_i}\frac{\bb(\vec{v}_i,\vec{v}^{*}_{j+1})}{\bb(\vec{v}^{*}_j,
    \vec{v}^{*}_{j+1})}\, \vec{v}^{*}_{j}.
\end{equation}

We see that Equation~\eqref{eqn:GSOisortho} holds for the resulting vector
$\vec{v}^{*}_i.$ Furthermore, if $\bb(\vec{v}^{*}_i, \vec{v}^{*}_i)
\neq 0$, we can immediately use Equation~\eqref{eqn:GSOnew} to compute
$\vec{v}^{*}_{i+1}$ using the set of bad indices $B_{i+1}=B_{i}.$

The case where $\bb(\vec{v}^{*}_i, \vec{v}^{*}_i)=0$ is taken care of
in Section~\ref{sec:ouralg} to ensure that Equation~\eqref{eqn:GSOnew} is only used
with a proper set of bad indices.
\\
We call the family $(\vec{v}^{*}_1,\cdots, \vec{v}^{*}_i)$ together
with the set $B_i$ the {\bf generalized Gram vectors at position $i$}
of $(\vec{v}_1,\cdots, \vec{v}_k)$. Note that for any $i'<i$ with
$i'\not\in B_i$, the truncated family
$(\vec{v}^{*}_1,\cdots, \vec{v}^{*}_{i'})$ and the set
$B_{i'}=B_i\cap[1\cdots i'-1]$ form the generalized Gram vectors at
position $i'$. However, we should be careful not to truncate at a bad
index, since that would lead to a truncated lattice with
Gram-determinant zero.\\

Consider any sublattice $L_i$ of $L$, generated by $\bb$ and the
truncated family of vectors $(\vec{v}_1,
\cdots, \vec{v}_i)$. After the usual GSO process, the Gram determinant $L_i$ is
easily computed as:
$$
\det{L_i}=\prod_{j=1}^{i}\bb(\vec{v}^{*}_j, \vec{v}^{*}_j).
$$

With our generalized GSO, the formula is updated to:
$$
\det{L_i}=\prod_{j\in G_i}\bb(\vec{v}^{*}_j,
\vec{v}^{*}_j) \cdot \prod_{j\in B_i}-\bb(\vec{v}^{*}_j, \vec{v}^{*}_{j+1})^2.
$$

It becomes clear in Section~\ref{sec:ouralg} that the generalized GSO provides a
funtionnal replacement of the standard GSO while allowing us to deal
with some isotropic vectors.

\subsection{Hyperbolic planes}\label{sec:Hyp-planes}
In the generalized GSO, the new situation we encounter happens when two
consecutive projected vectors $\vec{v}^{*}_i$ and $\vec{v}^{*}_{i+1}$
are isotropic but not orthogonal. This means that the 
Gram matrix of the corresponding $2\times 2$ projected lattice has the
form:
$$
G_{\alpha}=
\begin{pmatrix}
  0& \alpha\\
  \alpha&0
\end{pmatrix}.
$$

Such a (two-dimensional indefinite ) lattice is called an hyperbolic
plane. We distinguish two situations. In the first, we have either
$\vec{v}^{*}_i\neq \vec{v}_i$ or $\vec{v}^{*}_{i+1}\neq
\vec{v}_{i+1}$. As we see in Section~~\ref{sec:ouralg}, this means that $ \vec{v}_i$
(or  $\vec{v}_{i+1}$) can, in fact, be used to improve the lattice occurring
before position $i.$  This bypasses the hyperbolic
plane configuration.\\

In the second situation, $\vec{v}^{*}_i= \vec{v}_i$ and
$\vec{v}^{*}_{i+1}=\vec{v}_{i+1}.$ Thus $(\vec{v}_i, \vec{v}_{i+1})$
is an hyperbolic plane orthogonal to the lattice before it.
We can remark that $\vec{v}_i+\vec{v}_{i+1}$  and
$\vec{v}_i-\vec{v}_{i+1}$ are eigenvectors of $G_{\alpha}$. Their
respective ``squared-norms'' are $2\alpha$ and $-2\alpha$. So the lattice 
 spanned by them has Gram matrix:
$$
\begin{pmatrix}
  2\alpha &0 \\
  0&-2\alpha
\end{pmatrix},
$$
and corresponds to a sub-lattice of index $2$ of the
indefinite lattice spanned by $\vec{v}_i$ and $\vec{v}_{i+1}.$

We can further remark that the minimal absolute value of the 
``squared-norm'' of a non-isotropic vector in that plane is $2\alpha$. If we choose a basis
$(\vec{w}_i, \vec{w}_{i+1})$ of the two dimensional lattice with 
$|\Scal{\vec{w}_i}{\vec{w}_i}|=2|\alpha|$ then, after
orthogonalization, we have
$|\Scal{\vec{w}^{*}_{i+1}}{\vec{w}^{*}_{i+1}}|=|\alpha|/2,$ to preserve
the determinant. This is precisely the technique used
in~\cite{IvanyosSz96} to avoid
isotropic vectors and we see that it creates a drop by a factor of 4
between the ``squared-norms'' of the consecutive orthogonalised vectors.
\\

Since we do not want to accept such a gap, we simply choose to keep
hyperbolic planes in the lattice basis.  After all, it's Gram
matrix looks pretty good even if it is anti-diagonal instead of
diagonal.

\subsection{Admissible Lattice Gram matrices and local zeroes}\label{sec:admissible}
To reflect the two previous Sections, we insist in our lattice
reduction algorithm, at any point in its execution, only constructs
sub-families of vectors $(\vec{v}_1, \cdots, \vec{v}_k),$ with (partial)
Gram matrix that we deem
{\bf admissible}. This notion is defined by induction by saying that:
\begin{enumerate}
\item \label{case:empty} The induction is started by accepting the empty matrix $\emptyset$ of
  dimension $0$ as admissible. 
\item \label{case:normal} If $G$ is admissible of dimension $k-1$ then any block Gram
  matrix:
  $$
  H=\begin{pmatrix}
    G&\transp{\vec{s}} \\
    {\vec{s}}&t
  \end{pmatrix}
  $$
  with determinant $\det{H}\neq 0$ is an admissible Gram of
  dimension $k$.
\item \label{case:hyp} If $G$ is admissible of dimension $k-2$ then any block Gram
  matrix:
  $$
  H=\begin{pmatrix}
    G& \transp{\vec{0}}& 0\\
   {\vec{0}}&0&\alpha\\
    0&\alpha &0
  \end{pmatrix},
  $$
  with $\alpha\neq 0$ is an admissible Gram of
  dimension $k$. In that case, the last two dimensions form a
  hyperbolic plane fully orthogonal to the rest of the lattice.
\end{enumerate}

\begin{lemma}
Every admissible Gram matrix is invertible (which the convention that
the empty matrix is).
\end{lemma}
\begin{proof}
  This is easily checked by considering the determinant. For
  case~\eqref{case:normal}, we 
  have $\det{H}\neq 0$ by assumption. For case~\eqref{case:hyp}, we see that
  $\det{H}=-\alpha^2\det{G}\neq 0.$
\end{proof}

We define the {\bf local potential}  $\pot{H}$ of an admissible Gram
matrix $H$ by the following induction:
\begin{enumerate}
\item In case~\eqref{case:empty}, we let $\pot{\emptyset}=1.$
\item In case~\eqref{case:normal}, we let $\pot{H}=|\det(H)|\,\pot{G}.$
\item In case~\eqref{case:hyp}, we let $\pot{H}=|\alpha|^3\,\det(G)^2\,\pot{G}.$
\end{enumerate}
This notion of potential is helpful for the complexity analysis of our
algorithm, see Section~\ref{sec:proof}.\\

It is clear that the generalized GSO of Section~\ref{sec:extendedGSO} can
be applied to orthogonalize any extra vector against a given
admissible Gram matrix $G$, with the a set of bad indices pointing to the hyperbolic
planes inside of $G$. Assume that $G$ has dimension $\kappa$ and let
$(\vec{v}_1, \cdots, \vec{v}_\kappa)$ be the corresponding lattice
basis. More precisely, this means that for all $(i,j)$, we have:
$$
G_{i,j}=\bb(\vec{v}_i,\vec{v}_j).
$$

Let $\vec{w}$ be an arbitrary extra vector and let $\vec{w}^*$ denote
its orthogonalized vector obtained by the generalized GSO. Thanks to
Equation~\eqref{eqn:GSOnew}, we know that:
$$
\vec{w}-\vec{w}^*=\sum_{i=1}^{\kappa}\theta_i\vec{v}_i,
$$
for some coefficient vector $(\theta)_{i=1}^{\kappa}.$ Furthermore,
since $\vec{w}^*$ is orthogonal to all the vectors $\vec{v}_i$,  we
easily see that:
\begin{equation}
  \label{eq:thetavector}
\begin{pmatrix}
  \bb(\vec{w},\vec{v}_1)\\
  \bb(\vec{w},\vec{v}_2)\\
  \vdots\\
  \bb(\vec{w},\vec{v}_{\kappa})
\end{pmatrix}=
G_{i,j}\cdot
\begin{pmatrix}
  \theta_1\\
  \theta_2\\
  \vdots\\
  \theta_{\kappa}
\end{pmatrix}.
\end{equation}

When $\bb(\vec{w}^*, \vec{w}^*)=0,$ we say that the vector $\vec{w}$
{\bf adheres} (or is {\bf adherent}) to $G$. If $\vec{w}$ adheres to
$G$ and $(\theta)_{i=1}^{\kappa}$ is an integral vector, we say that
$\vec{w}$ is a {\bf $G$-zero}.

\begin{lemma}\label{lem:nonadh}
  If $G$ is an admissible Gram matrix with corresponding basis
  $(\vec{v}_1, \cdots, \vec{v}_\kappa)$ and $\vec{w}$ a vector that
  does not adhere to $G$, then $(\vec{v}_1, \cdots,
  \vec{v}_\kappa,\vec{w})$ spans an admissible Gram matrix $H$ of
  dimension $\kappa+1$, with Gram determinant:
  $$\det(H)=\det(G)\cdot \bb(\vec{w}^*, \vec{w}^*).$$
\end{lemma}
\begin{proof}
$H$ and its determinant are easily computed from the definition of adherence.
\end{proof}

\subsection{Dealing with $2\times 2$ blocks}
In the LLL algorithm, each (projected) $2\times 2$ block is reduced
using a simple variant of Gauss-Hermite-Lagrange algorithm induced by
the Lov{\`a}sz condition. We now recall the Gauss-Hermite-Lagrange
algorithm together with its variant and show the relationship with the
reduction of a positive definite binary quadratic form.

Let $\vec{u}$ and $\vec{v}$ denote two linearly independent
vectors. The geometry of the two dimensional lattice they generate is
fully described by the three numbers $N_1=\Norm{\vec{u}}^2$,
$N_2=\Norm{\vec{v}}^2$ and $S=\Scal{\vec{u}}{\vec{v}}.$ In our lattice
reduction application, $N_1$, $N_2$ and $S$ are rational
numbers\footnote{We leave the question of adapting floating point LLL
  techniques to the indefinite case open.}. We say that the basis
$(\vec{u},\vec{v})$ is Gauss-Hermite-Lagrange reduced if and only if:
\begin{align}
  N_1& \leq N_2 \quad\mbox{and}\\
 |S| & \leq N_1/2.
\end{align}
Given a positive real $\gamma_n<1$ we say that the basis is
$\gamma_n$-Lov{\`a}sz reduced if and only if:
\begin{align}
 \gamma_n\, N_1& < N_2 \quad\mbox{and}\\
 |S| & \leq N_1/2.
\end{align}

The algorithm is standard and consists of a sequence of elementary
reduction steps. In each step, we assume that $N_1<N_2$ and choose
$\lambda$ as the closest integer to $S/N_1$. We then compute a new
norm and scalar product:
\begin{align*}
  S'&=S-\lambda &\ \mbox{and}\\
  N'_2&= N_2-2\lambda\,S\lambda^2\,N_1.&
\end{align*}
If $N'_2$ is lower than $N_1$ (by a factor $\gamma_n$ for Lov{\`a}sz),
we exchange the norms and repeat.

To see the link with binary quadratic forms, let us define:
\begin{equation}
  Q(x,y)=\Norm{x\vec{u}+y\vec{v}}=N_1\,x^2+2\,S\,xy+N_2\,y^2.
\end{equation}
We can remark that $Q(x,y)$ is a binary quadratic form. Moreover,
since it is induced by a norm, it is positive definite.
Looking up the definition of reduction for a positive definite form
$a\,x^2+b\,xy+c\,y^2,$ we find the condition $|b|\leq a
\leq c.$ In our case, this is precisely the same as the
Gauss-Hermite-Lagrange reduction. As a consequence, in our algorithm
to reduce a lattice described by $N_1$, $N_2$ and $S$, we simply
apply a reduction step to the definite quadratic form $(N_1,2S,N_2)$. Since, this
complete reduction does not check the Lov{\`a}sz condition, we defer
this check to the main loop where we decide whether the reduction is
worth to apply or not. \\

Of course, this is equivalent to using the
above lattice algorithm directly. However, this point of view is
especially fruitful to treat the indefinite case.

\subsubsection{Quadratic indefinite forms}
Now, when reducing an indefinite lattice, the projected $2\times 2$
blocks are not necessarily positive definitive. In previous works, the
Lov\'asz condition was simply adapted by adding absolute values around
norms and, possibly, making special treatment for isotropic vectors. However,
this strongly differs from the notion of reduction that is usually
encountered when dealing with indefinite quadratic forms and that we
described in Section~\ref{sec:binforms}. Here, we correct this
discrepancy, thus benefiting from the theoretical framework of
indefinite binary quadratic forms. Note that we also need the
cases where $\Delta$ is square that are presented in
Section~\ref{sec:missing}, since nothing prevents them from happening
during the overall lattice reduction.

This modification is easily achieved by copying the approach we use for the
definite case. Simply compute $N_1=\Scal{\vec{u}}{\vec{u}}$,
$N_2=\Scal{\vec{v}}{\vec{v}}$ and $S=\Scal{\vec{u}}{\vec{v}}$ and
reduce the form $(N_1,2S,N_2)$. The main difference here is that after
reaching the first reduced form, we can continue reducing to find a
better one. To make sure that our lattice reduction algorithm remain
polynomial time, we bound the number of extra reduction steps we allow
by some constant. We keep the best option among these extra steps
before deciding whether the reduction step is worthwhile, i.e., we
have make significant local progress by applying it.

More precisely, when the initial $N_1$ is transformed into an updated
$N'_1$, we consider that the reduction step is worthwhile when:
$$
|N'_1|< \gamma_n\,|N_1|.
$$
This mimics the idea of $\gamma_n$-Lov{\`a}sz reduction in the
indefinite case.

\subsection{Dealing with hyperbolic planes}\label{sec:swaphyp}
While the previous Section deals with $2\times 2$ blocks of vectors,
it leaves open the case of the mutual reduction of two consecutive
hyperbolic planes or of a single hyperbolic plane and an extra
vector. Again, we take a positive real $\gamma_h< 1$ as parameter.

For two hyperbolic planes, due to the admissibility requirement, it
suffices to look at the case:
$$
\begin{pmatrix}
  0&\alpha &0 &0 \\
  \alpha &0 &0& 0\\
  0 & 0& 0&\beta \\
  0 & 0& \beta&0\\
\end{pmatrix}.
$$
If $\gamma_h\,|\alpha|\leq |\beta|$ we keep the two blocks in place. Otherwise,
we exchange their order.

When an hyperbolic plane is followed by a vector, up to scaling, it
suffices to consider the situation:
$$
\begin{pmatrix}
  0&1 &\alpha \\
  1 &0 &\beta\\
  \alpha& \beta& \gamma
\end{pmatrix}.
$$
If, after size reduction modulo $1$, we have $\alpha\neq 0$, we put
the vector before the hyperbolic plane. After that, the first vector
of the hyperbolic plane is no longer fully orthogonal to the previous
vectors and it can directly be incorporated into an admissible Gram matrix. If
$\alpha=0$ and $\beta\neq 0$, we exchange the two vectors of the
hyperbolic plane to recover the situation $\alpha\neq 0$.
\\

Finally, when $\alpha=\beta=0$, we reorder when $|\gamma|<\gamma_h$ and leave the
situation as it is otherwise. With a vector followed by an hyperbolic plane, we are in a similar
situation with the additional guarantee that $\alpha=\beta=0$. Thus, we
leave the situation as it is when $\gamma_h\,|\gamma|\leq 1$  and reorder when
$|\gamma|>\gamma_h^{-1}$.

We normally choose $\gamma_h=\gamma_n$, however, we also spend a
moment considering the choice $\gamma_h=1.$

\subsection{On size reduction}\label{sec:sizred}
Size reduction is an important component of LLL. Its goal is to make
sure that the coefficients of any of the currently considered lattice
basis vectors 
in the current GSO basis remain bounded. Without this control, we could
spiral into a basis satisfying all the Lov{\`a}sz conditions, thus
having a nicely bounded GSO, yet having unwieldy large
coefficients. By adding ``long distance'' control to the coefficients,
this situation is neatly avoiding.

With the ordinary lattices, the size reduction condition\footnote{This
condition can be slightly weakened when desired. This is typically the
case with floating point implementations of LLL.}:
$$
\forall j<i: \quad |\Scal{\vec{v}_i}{ \vec{v}^{*}_j}|\leq \Norm{\vec{v}^{*}_j}/2,
$$
looks pretty boring. Indeed, it looks just like extending the same
condition specialized to $j=i-1$ that forms one of the two Lov{\`a}sz
conditions. 

When turning to indefinite lattices, we still need a similar
condition. However, we enforce it with two modifications:
\begin{itemize}
  \item First, we only perform size reduction for indices
    $j<i-1$. For $j=i-1$, we can do size reduction when the projected
    block at positions$(j,i)$ is definite. If the block is indefinite,
    we replace size-reduction by a clean-up operation that is compatible
    with the reduction of indefinite blocks. Otherwise, we would
    destroy the reducedness of the block, thus entering an infinite loop.
    \item Second, the formula for size reduction needs to be adapted
      for the bad indices of generalized GSO. This is easily done.
\end{itemize}

\section{The new indefinite LLL}\label{sec:ouralg}
\subsection{Preliminary analysis: improvement of Simon's
  algorithm}\label{sec:prelimSimon}

Before giving our full algorithm for indefinite lattices, it is useful
to first follow an approach similar to Simon's and consider what happens if no
isotropic GSO vectors arise during the course of the lattice
reduction. As this is just a warm-up, we do not give an
algorithm. Instead, we assume that it outputs a lattice
with no isotropic GSO vectors and where every two-dimensional
projected block is reduced. The relevant difference here is that for
indefinite blocks, we use the notion of reduction coming from
quadratic form instead of the modified Lov{\'a}sz condition with
absolute values.

In this situation, we are interested by
the quality of the output basis and especially by relationship between
the value of the ``squared-norm'' of the first vector and the Gram-determinant of the
lattice.

For projected two-dimensional blocks, two situations arise:
\begin{itemize}
  \item The block at positions $(i,i+1)$ is positive/negative definite and we 
    have:
    \begin{equation} \label{case:definite}
      |\bb(\vec{v}^{*}_i, \vec{v}^{*}_i)|\leq C_{\mbox{LLL}}^2|\bb(\vec{v}^{*}_{i+1}, \vec{v}^{*}_{i+1})|.
    \end{equation}

  \item The block at positions $(i,i+1)$ is indefinite (and reduced) and we 
    have:
    \begin{equation}\label{case:indefinite}
    |\bb(\vec{v}^{*}_i, \vec{v}^{*}_i)|\leq |\bb(\vec{v}^{*}_{i+1}, \vec{v}^{*}_{i+1})|.
    \end{equation}
\end{itemize}
\vspace{0.5cm}

As a consequence, we obtain that:
    \begin{equation}\label{case:indefboundlocal}
    |\bb(\vec{v}_1, \vec{v}_1)|\leq C_{\mbox{LLL}}^{2N_i} |\bb(\vec{v}^{*}_{i+1}, \vec{v}^{*}_{i+1})|,
    \end{equation}
where $N_i$ denotes the number of definite blocks among the $i-1$
blocks in the range $[1\cdots i]$.
In turn, this implies that:
\begin{equation}\label{case:indefboundglobal}
  |\bb(\vec{v}_1, \vec{v}_1)|^{d}\leq
  C_{\mbox{LLL}}^{2\sum_{i=2}^{d}N_i} \cdot |\det{G_L}|.
\end{equation}

Since we want the upper bound on $ |\bb(\vec{v}_1, \vec{v}_1)|$ to
be as tight as possible, we would like $\sum_{i=2}^{d}N_i$ to be
small. In essence, this is similar to~\cite[Remark
1.5]{Simon05}. However, we want to quantify the potential improvement.

Ideally, if we could achieve $\sum_{i=2}^{d}N_i=0$, the bound from
Equation~\eqref{case:indefinite} would become extremely nice. However, this
is not possible for all lattices. In fact, the signature of the
lattice allows us to compute upper and lower bounds on
$\sum_{i=2}^{d}N_i$. Indeed, using Theorem~\ref{th:GSOEigen}, we can
relate the signature of the lattice to the number of positive and
negative diagonal entries in its Gram-GSO. Since we assume that we are not encountering any isotropic vector, we
can use ordinary GSO (or equivalently assume that the number of
hyperbolic planes is $0$). It is easy to remark that when two
consecutive diagonal entries in the Gram-GSO have the same sign, the
corresponding projected $2\times 2$ block is (positive or negative)
definite and when the signs are opposite the block is indefinite.\\

For simplicity, assume that the number of positive eigenvalue $n^{+}$ is not
smaller than the number of negative eigenvalues $n^{-}$. Then, the
configuration that minimizes $\sum_{i=2}^{d}N_i$ occurs when the
Gram-GSO starts with a positive entries and alternates signs until all
the $n^{-}$ negative entries have been used. In that scenario, we have
$2\,n^{-}$ consecutive indefinite blocks\footnote{Unless, $n^{-}=n^{+}$,
  in which case all of the $2\times 2$ blocks are indefinite. But
  there are only $2\,n^{-}-1$ total blocks in total.}. As a consequence,
in this situation, we achieve:
\begin{align*}
  \sum_{i=2}^{d}N_i &=
  \sum_{i=2\,(n^{-}+1)}^{d}i-2\,n^{-}-1=\frac{(d-2n^{-}-1)\cdot(d-2n^{-})}{2}\\
                    &=\frac{(n^{+}-n^{-}-1)\cdot(n^{+}-n^{-})}{2}
  =\frac{\sigma\cdot(\sigma-1)}{2},
\end{align*}
where $\sigma$ is the signature of the lattice.

As a consequence, the exponent of $C_{\mbox{LLL}}$ in
Equation~\eqref{case:indefboundglobal} can be as low as
$\sigma\cdot(\sigma-1)$ instead of the much larger $n\cdot(n-1)$ that
appears for ordinary lattices\footnote{For doing the comparison, remark that
  Equation~\eqref{case:indefboundglobal} is essentially the bound of
  Theorem~\ref{th:LLL} raised to the power $2n$.}.\\

At the other extreme, we maximize $\sum_{i=2}^{d}N_i$ when we have
$n^{+}$ positive entries followed by $n^{-}$ negative ones. In that
situation, when $i\leq n^{+}$ we have $N_i=i-1$ and when $i>n^{+}$ we
have $N_i=i-2$. Thus,
\begin{equation*}
  \sum_{i=2}^{d}N_i =\sum_{i=2}^{d}(i-1)-n^{-}
                    =\frac{d\cdot(d-1)}{2}-n^{-}
                    =\frac{d\cdot(d-2)+\sigma}{2}.
\end{equation*}

Thus, we potentially have a major gain when the sum is minimized.
Note that even when the sum is maximized, for an indefinite form we
necessarily have $\sigma\leq d-2$. This already improves the bound
from the last paragraph of~\cite[Theorem~1.4]{Simon05}.\\

As a direct consequence of the above, we see that one of the important
goals of our indefinite reduction algorithm is to favor sign
alternance in the Gram-GSO as much as feasible. Of course, in the 
general case, we also need to account for the contribution of hyperbolic
planes to the bound.

\subsection{Designing the algorithm}
Taking into account the elements we have gathered, our design goals
are the following:
\begin{enumerate}
\item Follow the overall structure of Algorithm~\ref{alg:LLL}.
\item Replace Gram-Schmidt orthogonalization by generalized GSO.
\item Design a local reduction strategy with the following subgoals:
  \begin{enumerate}
  \item Make use of the reduction of binary quadratic forms to deal
    with $2\times 2$ blocks.
  \item Aim at regrouping sign alternance in the GSO toward the small
    indices.
  \item Forbid isotropic vectors unless they appear as a pair forming
    a hyperbolic plane.
    \item Swap hyperbolic planes around as proposed in Section~\ref{sec:swaphyp}.
  \end{enumerate}
\item Modify size-reduction to avoid unreducing $2\times 2$ blocks.
\end{enumerate}

We give the high-level description of our algorithm as
Algorithm~\ref{alg:indefLLL}. Compared to Algorithm~\ref{alg:LLL}, we
change the meaning of the variable $k$. In LLL, the current $2\times
2$ block being considered is formed of the vectors at position $k$ and
$k+1$. In our algorithm, for the sake of clarity, it is preferable to
consider the block consisting of the vectors at position $k-1$ and $k$.\\

\begin{algorithm}
\DontPrintSemicolon
\SetKw{loopbreak}{break}
\SetKw{continue}{continue}
\KwData{Description of input lattice and parameter $\gamma_0$
  (typically $0.99$) }
\KwResult{Unimodular $U$ that reduces the input lattice}
\tcp{All operations implicitly updates $U$ and $L_c$ as necessary}
Set {\it Current position} $k$ to $1$\;
Copy input $L$ to current lattice $L_c$ (with vectors $\vec{v}_i$)\;
\While{$k \leq \dim(L_c)$}{
  {\bf Let} $G$ be the  Gram matrix of $(\vec{v}_1,\cdots,\vec{v}_{k-1})$\;
  \For{$i\in [k\cdots \dim(L_c)]$}{
    {\bf Orthogonalize and size reduce} $\vec{v}_i$ with respect to $G$\; 
    \lIf{$\vec{v}_i$ is not a $G$-zero}{{\bf Report} position $i$ and \loopbreak}
    \lIf{$i>k$ and $\bb(\vec{v}_{i-1}, \vec{v}_{i})\neq 0$}{{\bf Report} pair $(i-1,i)$ and \loopbreak}
  }
  \uIf{Nothing reported}{{\bf Find} and {\bf Report} a pair $(i,j)$ such
    $\bb(\vec{v}_{i}, \vec{v}_{j})\neq 0\ \ $\tcc*[h]{$k\leq i<j$}}
  \lIf{Nothing reported}{{\bf Return} $U$ and
    $L_c$  --- {\bf Terminate} algorithm.}
  \uIf{Position $i$ reported}{
      {\bf Cyclic shift} vector in position $i$ to position $k$\;
      \lIf{$k=1$}{Let $k=2$ and \continue}
     \uIf{$(k-2,k-1)$ is hyperbolic}{{\bf Apply plane reduction} to
       hyperbolic plane and vector at $k$}
     \Else{{\bf Apply vector reduction} to vector at position $k$\;
     }
     }
    \Else(\tcc*[h]{Pair $(i,j)$ reported}){
      {\bf Cyclic shift} vector in position $i$ to position $k$\; 
      {\bf Cyclic shift} vector in position $j$ to position $k+1$\;
      \lIf{$k=1$}{Let $k=3$ and \continue}
      {\bf Apply plane reduction} to hyperbolic plane
      $(k,k+1)$ and vector~$k-1$\;
    }
      {\bf Change} $k$ as dictated by reduction and let $k=\max(1,k)$\; 
    }
    {\bf Return} $U$ and $L_c$\;
\caption{\label{alg:indefLLL}  High-level description of indefinite LLL}
\end{algorithm}

Most of the various subroutines are easy to write from our earlier
section and we do not write them explicitly. Here is a quick
mapping of the correspondances:
\begin{itemize}
\item {\bf Orthogonalize} refers to the generalized GSO of
  Section~\ref{sec:extendedGSO}. As noted in
  Section~\ref{sec:admissible} it can be applied here if $G$ is
  admissible. We check this in the proof Section~\ref{sec:proof}.
\item {\bf Size reduce} $\vec{v}_i$ with respect to $G$ refers to
  Section~\ref{sec:sizred} and performs both the size-reduction and
  clean-up of that section. More precisely, it anticipates that
  $\vec{v}_i$ will be moved to position $k$. Because of this
  $\vec{v}_i$ is size-reduced relatively to every $\vec{v}_j$ with
  $j<k-1$. For $\vec{v}_{k-1}$, we use size reduction if (after the
  move) the projected block is definite or degenerate and clean-up if
  it is indefinite. We give the pseudo-code of the clean-up operation
  as Algorithm~\ref{alg:cleanup}
\item {\bf Apply plane reduction} refers to the description provided
  in Section~\ref{sec:swaphyp}, using either $\gamma_h=\gamma_0$ or
  $\gamma_h=1$. If a swap is performed, we return the indication to
  move $k$ so that $k-1$ points before the vector or block that has
  been moved back.
  
\item {\bf Cyclic shift} of vectors, giving position $i$ and $k$
  simply replaces the vectors $(\vec{v}_k,\cdots, \vec{v}_{i})$ by
  $(\vec{v}_{i},\vec{v}_k,\cdots, \vec{v}_{i-1}).$
\end{itemize}
For more details about these operations, the reader is invited to read
the Magma code provided in Appendix~\ref{sec:magmacode}. Note that the
routine  {\bf Apply plane reduction} is inlined and duplicated in the code.

Concerning {\bf Apply vector reduction}, the situation is more subtle.
Indeed, we need a strategy to enforce sign alternance. We provide two
versions, a first one with no sign alternance strategy is given as
Algorithm~\ref{alg:red-nosign} and another one with a sign alternance
strategy in Algorithm~\ref{alg:red-sign}. In the first version, we simply aim at
decreasing the first vector in any $2\times 2$ block by a factor of at
least $\gamma_0$, while avoiding to keep any $G$-zero. We also use the
ability to continue reduction beyond the first reduced form for
indefinite blocks to get better vectors. In the second version, we
favor sign alternance. More precisely, after the block first reaches
reducedness, if  the sign alternance is already present, we only
perform vector improvements that keep it in place. When the block is
not reduced on input, we look at the first reduced form $(a,b,c)$ we reach. If
$a$ has the right sign, we keep it. If not, since $c$ has a different
sign from $a$, we check whether $(c,b,a)$
is satisfactory. More precisely, we ask for $|c|$ to be
$\gamma_0$-short compared to the first input ``squared-norm'' and we
further insist that $4c^2\leq b^2-4ac$. Equivalently, this means that
the corresponding ``squared-norm'' is shorter that the square root of
the absolute value of the $2\times 2$ projected lattice determinant.

\begin{algorithm}
\DontPrintSemicolon
\SetKw{loopbreak}{break}
\SetKw{continue}{continue}
\KwData{Initial sublattice $(\vec{v}_1,\cdots, \vec{v}_{\kappa})$ and
  extra vector $\vec{v}_i$.}
\KwResult{Cleaned-up $\vec{v}_i$ }
Let $N_1=\bb(\vec{v}_{\kappa}^*, \vec{v}_{\kappa}^*)$\;  
Let $S=\bb(\vec{v}_{\kappa}^*, \vec{v}_{i})$\; 
Let $N_2=\bb(\vec{v}_{i}^*, \vec{v}_{i}^*)+S^2/N_1$\;
Let $\Delta=S^2-N_1*N_2$\;
\lIf{$\Delta\leq 0$}{Call ordinary size-reduction}
\uIf{$S=0$ and $N_1+N_2=0$}{Let $\lambda=0$}
\uElseIf{$N_2\neq 0$ or $|S|\neq \sqrt{\Delta}$ or $|N_1|\neq
  \sqrt{\Delta}$}{
  \uIf{$|N_1|>\sqrt{\Delta}$}{Let $\lambda=\lfloor -S/N_1\rceil$}
  \Else {
    Let $\lambda=(\sqrt{\Delta}-S)/N_1$\;
    \leIf(\tcc*[h]{Rnd towards $0$}){$N_1>0$}{Let $\lambda=\lfloor \lambda\rfloor$}{Let $\lambda=\lceil \lambda\rceil$}
    }
  }
  \Else(\tcc*[h]{$N_2= 0$ and $|S|= \sqrt{\Delta}$ and $|N_1|=
  \sqrt{\Delta}$}){
    Let $\lambda=\lfloor -S/N_1\rceil$
    }
Let $\vec{v}_i=\vec{v}_i+\lambda\,\vec{v}_{\kappa}$\;
\caption{\label{alg:cleanup}  Clean-Up part of size-reduction}
\end{algorithm}

\begin{algorithm}
\DontPrintSemicolon
\SetKw{loopbreak}{break}
\SetKw{continue}{continue}
\SetKwInput{KwParam}{Parameters}

\KwData{Initial sublattice $(\vec{v}_1,\cdots, \vec{v}_{k-1})$ and
  next vector $\vec{v}_k$}
\KwParam{ $\gamma_0<1$ and max extra rounds count $M$.}
\KwResult{Next change for $k$}
Let $G$ be the Gram matrix of  $(\vec{v}_1,\cdots, \vec{v}_{k-2})$
Let $N_1=\bb(\vec{v}_{k-1}^*, \vec{v}_{k-1}^*)$\;  
Let $S=\bb(\vec{v}_{k-1}^*, \vec{v}_{k})$\; 
Let $N_2=\bb(\vec{v}_{k}^*, \vec{v}_{k}^*)+S^2/N_1$\;
Let $\Delta=S^2-N_1*N_2$\;
\lIf{$\Delta\leq 0$}{Call $\gamma_0$-Lov{\'a}sz reduction}
Let $a=N_1$, $b=2\,S$ and $c=N_2$\;
\uIf{$(a,b,c)$ is not already reduced}{Apply reduction steps until
  $(a',b',c')$ is reduced or $|a'|\leq |a|/2$\;
  \lIf{$|a'|<\gamma_0\,|a|$}{Mark {\bf worthwhile}}
  \lIf{$|a'|\leq|a|$ and $c=0$ and $c'\neq 0$}{Mark {\bf worthwhile}}
}
\Else{Apply up to $M$ reduction steps with intermediates $(a',b',c')$
  \;
    \lIf{$|a'|<\gamma_0\,|a|$}{Mark {\bf worthwhile} and stop the
      reduction loop}
  }
  \lIf{Transformation not worthwhile}{{\bf Return} $k$-change$=+1$}
  Let $\vec{w}_{k-1}$ and $\vec{w}_{k}$ be the new vectors after
  applying the transform\;
  {\bf Size reduce} $\vec{w}_{k-1}$ with respect to $G$\;
  \If{$\vec{w}_{k-1}$ is not a $G$-zero and  $(\vec{w}_{k-1},\vec{w}_{k})\neq (\pm\vec{v}_{k-1},\pm\vec{v}_{k})$}
    {{\bf Replace} $(\vec{v}_{k-1},\vec{v}_{k})$ by
    $(\vec{w}_{k-1},\vec{w}_{k})$ and {\bf Return} $k$-change$=-1$}
  \If{$\vec{w}_{k-1}$ is a $G$-zero and  $(\vec{w}_{k},\vec{w}_{k-1})\neq (\pm\vec{v}_{k-1},\pm\vec{v}_{k})$}
    {{\bf Replace} $(\vec{v}_{k-1},\vec{v}_{k})$ by
      $(\vec{w}_{k},\vec{w}_{k-1})$ and {\bf Return} $k$-change$=-1$}
    {\bf Return} $k$-change$=+1$ \tcc*[h]{Fell through the tests}
  \caption{\label{alg:red-nosign}  Vector reduction against previous (no
sign strategy)}
\end{algorithm}

\begin{algorithm}
\DontPrintSemicolon
\SetKw{loopbreak}{break}
\SetKw{continue}{continue}
\tcc{ Replacement for lines 7--14 of Algorithm~\ref{alg:red-nosign}}
Determine position $\ell$, right before $k-1$ skipping any hyperbolic
plane\;
\lIf{$\ell=0$}{Use Algorithm~\ref{alg:red-nosign}}
\uIf{$(a,b,c)$ is not already reduced}{Apply reduction steps until
  $(a',b',c')$ is reduced or \\$|a'|\leq |a|/2$ and $a'$ has the right
  sign (either $a'=0$ or opposite to $\vec{v}_{\ell}^{*}$) \;
  \If{$a$ of wrong sign and $c'\neq 0$ and $|c'|<\gamma_0\,|a|$ and
    $c'\leq \sqrt{\Delta}$}{Swap $a'$ and $c'$}
  \lIf{$|a'|<\gamma_0\,|a|$}{Mark {\bf worthwhile}}
  \lIf{$|a'|\leq|a|$ and $c=0$ and $c'\neq 0$}{Mark {\bf worthwhile}}
}
\Else{Apply up to $M$ reduction steps with intermediates $(a',b',c')$
  \;
    \If{$a'$ has right sign (or is zero) and $|a'|<\gamma_0\,|a|$}{Mark {\bf worthwhile} and stop the
      reduction loop}
  }
  \caption{\label{alg:red-sign}  Vector reduction with
sign strategy.}
\end{algorithm}

\subsection{Proofs: complexity and approximation
  factor} \label{sec:proof}

In this Section, we want to prove our main theorem, as stated
below.
\begin{theorem}\label{th:IndefLLL}
  For any lattice $L$ of dimension $d$ of rank $\ell$, Algorithm~\ref{alg:indefLLL}
  outputs, in polynomial time, a basis $(\vec{v}_1,
  \cdots,\vec{v}_d),$ such that:
  \begin{itemize}
    \item If $\vec{v}_1$ is non-isotropic then:
      $$
      |\bb(\vec{v}_1, \vec{v}_1)|^{\ell}\leq
      \gamma_0^{-\ell(\ell-1)}\,D_{LLL}^{\sum_{i=1}^{\ell}N_i}\,|\det_{\neq 0}(G_L)|.
      $$
    \item  If $\vec{v}_1$ is isotropic then $(\vec{v}_1, \vec{v}_2)$
      is a hyperbolic plane and:
            $$
      |\bb(\vec{v}_1, \vec{v}_2)|^{\ell}\leq
      \gamma_0^{-\ell(\ell-1)}\,D_{LLL}^{\sum_{i=1}^{\ell}N_i}\,|\det_{\neq
        0}(G_L)|.
      $$
    \end{itemize}
  Where $\gamma_{0}=1-\epsilon$ with $\epsilon>0$  is a parameter of
  the algorithm and 
  $
  D_{LLL}=\frac{\gamma_0^2}{\gamma_0-0.25}.
  $

  As before, $N_i$ stands for the number of definite $2\times 2$
  projected blocks up to position~$i$.
  Furthermore, for all pairs $(i,j)$ with $j>\ell$ we have:
  $$
  \bb(\vec{v}_i, \vec{v}_j)=0.
  $$
\end{theorem}

Note that the seemingly weird choice of constant $D_{LLL}$ is made to
simplify the expression of the exponent of $\gamma_0$.\\

Before proving this main theorem, we first study some properties of
our algorithm. As a preliminary, let us recall that
the runtime of LLL is usually bounded by showing that an important quantity
called the {\bf potential} of the lattice decreases multiplicatively
throughout the algorithm. To analyze Algorithm~\ref{alg:indefLLL}, we
follow the same general strategy, however, this is somewhat trickier
since we only define our potential in Section~\ref{sec:admissible} for
admissible lattices.  In order to understand the running time, we
start by stating an important fact about Algorithm~\ref{alg:indefLLL}.

\begin{lemma}
 Whenever we enter the main loop of  Algorithm~\ref{alg:indefLLL}, the sublattice
 formed of the vectors $\vec{v}_1$, \dots, $\vec{v}_{k-1}$ is admissible.
\end{lemma}
\begin{proof}
  This clear when $k=1$ since the empty lattice is admissible.

  After $k$ decreases, the only way to encounter a non-admissible
  lattice would be to have the new $\vec{v}_{k-1}$ as part of a
  hyperbolic plane. This cannot happen because when decreasing $k$,
  {\bf Apply plane reduction} always points before any hyperbolic
  plane it has been processing.

  When $k$ is planned to increase, it means the reduction we just
  performed did not perform any change in the previous vectors. And,
  we have two possibilities.
  If we found a new  vector $\vec{v}_{k}$, it necessary have
  $\vec{v}^{*}_{k}\neq 0$, this leads to an admissible lattice of
  dimension $k$. If we found a hyperbolic pair $(\vec{v}_{k},
  \vec{v}_{k+1})$, we get an admissible lattice of
  dimension $k+1$.
\end{proof}

\begin{lemma}\label{lem:fulldet}
  As an immediate corollary, when  Algorithm~\ref{alg:indefLLL},  the
  lattice basis $(\vec{v}_1$, \dots, $\vec{v}_{\ell})$ (excluding the
  vectors from $V_0$ at the end of the basis if any)  with Gram basis
  $G_\ell$ is admissible.
  Furthermore, its Gram determinant is:
  $$
  \det(G_\ell)=\det_{\neq 0}(G_L)=\prod_{\begin{array}{c}i=1\\ \bb(\vec{v}^{*}_i, \vec{v}^{*}_i)\neq
    0\end{array}}^{\ell}\bb(\vec{v}^{*}_i, \vec{v}^{*}_i)\, \cdot\, \prod_{\begin{array}{c}i=1\\
    (i,i+1)\ \mbox{hyperbolic}\end{array}}^{\ell}-\bb(\vec{v}^{*}_i, \vec{v}^{*}_{i+1})^2.
  $$
\end{lemma}

In addition, we easy see that when the algorithm terminates, all the
vectors that are left unprocessed are necessarily
$G_\ell$-zeroes. Furthermore, since no hyperbolic planes are left they
are mutually orthogonal. This proves the last assertion of Theorem~\ref{th:IndefLLL}.

\subsubsection{Adherent vectors}
The first thing we we need to study is the behavior of the algorithm
after its finds and  reports a vector $\vec{v}_k$ that adheres 
to the sub-lattice $(\vec{v}_1,\cdots,\vec{v}_{k-1})$ with partial Gram
matrix $G$. Note that since the vector $\vec{v}_k$ is reported as
found, it is not a $G$-zero.\\

If $\vec{v}_k$ follows a hyperbolic plane, the
reduction step necessarily moves it before that plane.
If we had either $\bb(\vec{v}_k, \vec{v}_{k-1})\neq 0$ or
$\bb(\vec{v}_k, \vec{v}_{k-2})\neq 0$, after the move, the updated
Gram vector is no-longer isotropic and the hyperbolic plane is
destroyed as seen in Section~\ref{sec:swaphyp}. In fact,  in the case
$\bb(\vec{v}_k, \vec{v}_{k-2})\neq 0$, the local
three-dimensional projected Gram is, up to scaling and using the
notation of Section~\ref{sec:swaphyp}:
$$
G_3=\begin{pmatrix}
  0 &\alpha & \beta\\
  \alpha & 0 & 1\\
  \beta & 1 & 0
\end{pmatrix}.
$$
So we are back to a similar situation with a smaller hyperbolic plane
and an adherent vector again.

We now look at the case where $\vec{v}_k$ follows a vector with a
non-isotropic $\vec{v}^{*}_{k-1}$. Since $\vec{v}^{*}_k$ is isotropic,
we can compute  the two-dimensional Gram matrix of the projected
lattice coming from $\vec{v}_{k-1}$ and $\vec{v}_{k}$
by the following formula:
$$
G_2=\begin{pmatrix}
  \bb(\vec{v}^{*}_{k-1}, \vec{v}^{*}_{k-1}) & \bb(\vec{v}^{*}_{k-1},
  \vec{v}_{k})\\
  \bb(\vec{v}^{*}_{k-1},\vec{v}_{k}) &\frac{\bb(\vec{v}^{*}_{k-1},
  \vec{v}_{k})^2}{\bb(\vec{v}^{*}_{k-1}, \vec{v}^{*}_{k-1})}
\end{pmatrix}=
\begin{pmatrix}
  N & S\\
  S & S^2/N
\end{pmatrix},
$$
if we let $N=\bb(\vec{v}^{*}_{k-1}, \vec{v}^{*}_{k-1})$ and
$S=\bb(\vec{v}^{*}_{k-1},\vec{v}_{k}).$

Since $G_2$ has determinant $0$, its computes GCD of the two projected
vectors and thus outputs an updated pair of projected vectors
$(\vec{0}, \vec{v}^{*}_{k-1}/h$ for some integer $h$. When the
transformation is lifted, we get two new vectors $(\vec{w}_{k-1},
\vec{w}_{k})$. It is easy to see that $\vec{w}_{k-1}$ adheres to the
sublattice $(\vec{v}_1,\cdots,\vec{v}_{k-2})$ and  that $\vec{w}_{k}$
does not. If $\vec{w}_{k-1}$ is a zero relatively to the truncated
sublattice, it is moved away during the next iteration of the
algorithm. When this happens, $\vec{w}_{k}$ is reported as found and 
 $(\vec{v}_1,\cdots,\vec{v}_{k-2},\vec{w}_{k})$ is a new admissible
 lattice with Gram-determinant equal to the Gram-determinant of the initial
 $(\vec{v}_1,\cdots,\vec{v}_{k-2},\vec{v}_{k-1})$ divided by $h^2$.

 This implies the following lemma.
 \begin{lemma}
   When Algorithm~\ref{alg:indefLLL} encounters a vector $\vec{v}_k$
   that adheres to the admissible sub-lattice $(\vec{v}_1,\cdots,\vec{v}_{k-1})$,
   it creates (in polynomial time) a new admissible sub-lattice
   $(\vec{v}'_1,\cdots,\vec{v}'_{k-1})$. Furthermore, the determinant
   and potential of the new sub-lattice are 
   smaller than the initial values by a factor $4$ at least.
\end{lemma}

\subsubsection{Beyond adherent vectors}
As we see above, adherent vectors are only a temporary inconvenience
that can be dealt with a polynomial time overhead. We now want to
bound the number of backward steps that can occur throughout the
algorithm.

Let us denote by $T_K$ the minimum of the first time where a loop with
current position $k=K$ discovers a non-adherent vector and the first
time where a  loop with current position $k=K$ discovers a hyperbolic
plane. If $k$ becomes larger than $K$ before on of these events
happen, we let $T_K=\infty$. This happens when the loop at position
$k=K$ discover a hyperbolic plane first, thus setting $T_{K+1}$. At
time $T_K$, we have an admissible sublattice of dimension $K$. For
notational convenience, we denote by $T^{(f)}_K$ the minimum of
$T_{K+1}$, $T_{K+2}$ and the last loop of the algorithm. It
corresponds to first moment in time where we move beyond dimension
$K$, either constructing a larger admissible lattice or terminating
the algorithm.

\begin{lemma}
  At any point in time in the interval $[T_K,T^{(f)}_K[$ except during
  the integration phase of an adherent vector discovered at position
  $k=K+1$, the sub-lattice $(\vec{v}_1,\cdots,\vec{v}_{K})$ is
  admissible. As a consequence, its potential remains defined
  throughout that time interval.
\end{lemma}
\begin{proof}
  Except during the integration phase, all operations simply change
  the basis of the current sublattice of dimension $k-1$ with $k\leq
  K+1$. Furthermore, for any vector $\vec{v}_\ell$ with $k\leq \ell
  \leq K$, its generalized Gram-Schmidt $\vec{v}^{*}_\ell$ remains
  identical. So any previous vector or previous hyperbolic plane is
  found again when its position is reached.

  During an integration phase, things are temporarily perturbed. But
  as we saw, they fall back into place when the integration finishes,
  with a decreased potential.
\end{proof}

We define the {\bf current potential} of the lattice reduction process
as the potential of the largest admissible sub-lattice currently
discovered. During the algorithm, the number of times we change
the local potential that the current potential is pointing to is bounded by the
dimension. Since, each of these potentials is polynomially bounded in
the dimension and the size of the entries of the Gram matrix of the
input lattice, any of them can only be divided by $\gamma_0$ a
polynomial number of times before reaching its lower bound of $1$.

Thus, to prove that the overall algorithm performs a polynomial time
of arithmetic operations, we
simply need to show that at most polynomially many operations separate
two consecutive multiplication of a potential by $\gamma_0$. In the
next two subsections, we check that for the reduction steps of
non-adherent vectors and of hyperbolic planes. As in the classical runtime
analysis of LLL, we only bound the number of backward steps, since
a bound on the number of forward steps is easily derived from that.

From the bound on the number of arithmetic operations, we can derive
the polynomial time complexity by remarking that, as in LLL, if the
absolute value of largest entry in the input matrix is $B$, then all
numbers throughout the algorithm can be bounded by $(dB)^{d^2}.$

\subsubsection{Non adherent vectors} 
When a non-adherent vector is found, the subsequent reduction step can
either decrease $k$ or increase it.  Most of the time when it decreases the
Gram-determinant and potential of the sublattice of dimension $k-1$ is
reduced by a factor at least $\gamma_0$. There is the exception of
line 10 in Algorithm~\ref{alg:red-nosign}
and~\ref{alg:red-sign}. However, this special case never increases the
potential and we can check that it cannot occur often enough to more
than double the overall runtime.

\subsubsection{Hyperbolic planes}
For hyperbolic planes, if we set $\gamma_h=\gamma_0$, then the
potential decreases by $\gamma_0$ at least whenever a backward step
occurs.

If we set $\gamma_h=1$, the situation is a bit different. We need to
look more precisely at each type of exchange in order to bound them.
We have three cases.
\begin{itemize}
\item The exchange occurs between two hyperbolic plane. 
\item The exchange occurs between a hyperbolic plane and a vector with
  $\alpha=\beta=0$ and $|\gamma|<1$ (following the notations of
  Section~\ref{sec:swaphyp}. 
\item The exchange occurs because $\alpha\neq 0$ or $\beta\neq 0$.
\end{itemize}
In the first two cases, we are simply reordering small dimensional
lattices that are mutually orthogonal. The hyperbolic plane initially
at position $(k,k+1)$ can be moved back, possibly all the way to
$(1,2)$. Yet, it is a simple reordering that needs less than $k$
steps. As a consequence, even it this does not decrease the potential,
it only implies a multiplicative polynomial overhead.

In the last case, the local Gram matrix is first transformed into:
$$
\begin{pmatrix}
  \gamma &\alpha & \beta\\
  \alpha & 0 & 1\\
  \beta & 1 & 0
\end{pmatrix}.
$$
The local potential of this matrix is $\gamma^2\,\alpha^2$, with
$|\gamma|<1$ and $|\alpha|\leq 1/2$ (by size reduction). This is
obviously more than a factor $\gamma_0$ better than the potential
before the transformation (it was equal to $|\gamma|$).

As a consequence, setting $\gamma_h=1$ makes the analysis more complex
but still guarantees a polynomial runtime.

\subsubsection{Quality of approximation}
For any definite block consisting of vector at position $i$ and $i+1$,
we have the usual bound: 
$$
|\bb(\vec{v}_i^{*}, \vec{v}_i^{*})|\leq (\gamma_0-(1/4))^{-1}\, |\bb(\vec{v}_{i+1}^{*}, \vec{v}_{i+1}^{*})|.
$$
that we rewrite as
$$
\gamma_0^2\, |\bb(\vec{v}_i^{*}, \vec{v}_i^{*})|\leq D_{LLL}\, |\bb(\vec{v}_{i+1}^{*}, \vec{v}_{i+1}^{*})|.
$$

For indefinite blocks, we have two options depending on the test on
line~9 of Algorithms~\ref{alg:red-nosign} and~\ref{alg:red-sign}. If
this test marks the first reduced form we encounter as worthwhile,
then the corresponding $a$ is bounded in absolute value by the
determinant. If not, then the previous $a$ was already bounded in
absolute value by the determinant over $\gamma_0.$
As a consequence, we have:
$$
\gamma_0^2\, |\bb(\vec{v}_i^{*}, \vec{v}_i^{*})|\leq |\bb(\vec{v}_{i+1}^{*}, \vec{v}_{i+1}^{*})|.
$$
For a non-isotropic vector followed by a hyperbolic plane, we have:
$$
\gamma_0\, |\bb(\vec{v}_i^{*}, \vec{v}_i^{*})|\leq |\bb(\vec{v}_{i+1}^{*}, \vec{v}_{i+2}^{*})|.
$$
Similarly, with a hyperbolic plane, we have:
$$
\gamma_0\, |\bb(\vec{v}_i^{*}, \vec{v}_{i+1}^{*})|\leq |\bb(\vec{v}_{i+2}^{*}, \vec{v}_{i+2}^{*})|.
$$
And for two hyperbolic planes:
$$
\gamma_0\, |\bb(\vec{v}_i^{*}, \vec{v}_{i+1}^{*})|\leq |\bb(\vec{v}_{i+2}^{*}, \vec{v}_{i+3}^{*})|.
$$
Composing and multiplying the inequalities as before, with
multiplicity for the hyperbolic planes, we find:
$$
\gamma_0^{d(d-1)}|\bb(\vec{v}_1, \vec{v}_s)|^{d}\leq D_{LLL}^{\sum_{i=1}^{d}N_i}\,|\det(G_L)|.
$$
Where $s=1$ if the basis starts with a non-isotropic vector and $s=2$ otherwise.
This concludes the proof of Theorem~\ref{th:IndefLLL}.
\\

\begin{remark}
  The presence of hyperbolic planes improves Theorem~\ref{th:IndefLLL}
  in two ways. First, it reduces the exponent of $\gamma_0$ if we
  more carefully analyze it. Second, hyperbolic planes may also reduce
  the value of $\sum_{i=1}^{d}N_i$ below the lower bound of
  $\sigma(\sigma-1)/2$ from Section~\ref{sec:prelimSimon}.
  For example, consider the following matrix of dimension 3:
  $$
  \begin{pmatrix}
    1 & 0 & 0 \\
    0& 0 & 1 \\
    0 & 1& 0 \\
  \end{pmatrix},
  $$
  its signature is $1$ but $\sum_{i=1}^{d}N_i=0$. If we assemble
  diagonally $n$ copies of it, we obtain a matrix a dimension $3n$
  with signature $n$, that satisfies $\sum_{i=1}^{d}N_i=0$.
\end{remark}

\subsection{Heuristic expectation for the Algorithm}\label{sec:heuri}
Looking at the statement of Theorem~\ref{th:IndefLLL}, it is quite
clear that it misses a good upper bound on
$\sum_{i=1}^{d}N_i$. Unfortunately, it is possible to assemble a
counter-example family of reduced lattices that match the upper bound from
Section~\ref{sec:prelimSimon}. Furthermore these lattices are already
reduced for Algorithm~\ref{alg:indefLLL} and match the bound of
Theorem~\ref{th:IndefLLL} (with $\epsilon$ set to $0$).

To construct this family, we assemble definite blocks with Gram matrix
equal, up to scaling, to:
$$
\begin{pmatrix}
  2 & 1\\
  1 & 2
\end{pmatrix}.
$$
Each such block, as Gram-Schmidt squared-norms equal to $2$ and $3/2$,
thus matching the $4/3$ gap. These blocks can be assembled to form large LLL-reduced Gram matrices
of the form:
$$
G^{d}_\lambda=\begin{pmatrix}
  2 \lambda & \lambda & 0 &0 & \cdots & 0\\
  \lambda & 2\lambda & (3/4)\,\lambda &0 & \cdots & 0\\
  0& (3/4)\,\lambda & (3/2)\,\lambda  & (3/4)^2\,\lambda  &\cdots & 0\\
  0 & 0 &(3/4)^2\,\lambda & 2\times (3/4)^2\,\lambda &\cdots & 0 \\
  \vdots & \vdots&\vdots&\vdots &\ddots &\vdots \\
  0 &0 &\cdots&0&0& 2\times (3/4)^{d-2}\,\lambda& (3/4)^{d-1}\,\lambda\\
  0 &0 &\cdots&0&0& (3/4)^{d-1}\,\lambda&2\times (3/4)^{d-1}\,\lambda
\end{pmatrix}.
$$
Putting together two such matrices, we create the indefinite Gram
matrix:
$$H=\begin{pmatrix}
  G^{d}_\lambda &0\\
  0&-G^{d}_\mu
  \end{pmatrix},
$$
with $\lambda=16^{d-1}$ and $\mu=12^{d-1}$.
This ensures that $H$ is integral and the single indefinite block in
the middle is a scaled-up version of:
$$\begin{pmatrix}
  1&0\\
  0&-1
\end{pmatrix}.
$$
Note that the signature of $H$ is $\sigma=0$.
\\

As a consequence of this example, we cannot expect a worst case
improvement for the quality of the result of Algorithm~\ref{alg:indefLLL} . However,
for the regular LLL, we have a gap between the best proven bound and
the observed results. We expect a similar behavior here and we
conjecture the following heuristic.

\begin{heuristic}\label{th:heurqual}
  There exist a parameter $\epsilon$ and universal constants $c$ and
  $C$ such that:
  For a sufficiently random input bases of an indefinite lattice $L$,
  with dimension $d$ and signature $\sigma$,
  Algorithm~\ref{alg:indefLLL} outputs a basis that satisfies either:
  \begin{align*}
   |\bb(\vec{v}_1,\vec{v}_1)|^{d} &\leq C\cdot 
                                    (4/3)^{c\cdot\sigma(\sigma-1)}\,|\det(G_L)|
                                    \quad \mbox{or}\\
    |\bb(\vec{v}_1,\vec{v}_2)|^{d} &\leq C\cdot 
                                     (4/3)^{c\cdot\sigma(\sigma-1)}\,|\det(G_L)|,
  \end{align*}
  depending on whether the first vector is isotropic.
    
\end{heuristic}

One approach to back up the heuristic is to argue that the sign strategy
almost guarantees perfect sign alternance. Note that it would be the
case under the assumption that each indefinite block can be reduced
into a form that satisfies the sign alternance
criteria. Unfortunately, there exist indefinite blocks that require a
degradation of the potential to obtain sign alternance. For example,
$$\begin{pmatrix}
  1&3\\
  3&-6
\end{pmatrix}
$$
belongs to a reduction cycle of length $2$, which simply exchanges the
$1$ and $-6$. As a consequence, this block does not accept a reduced
form with a negative first vector small the absolute value of the
determinant (equal to $15$).\\

However, it is reasonable to argue that such example should be
rare. Indeed, for a quadratic real field $\QQ[\sqrt{\Delta}]$, the
Cohen-Lenstra heuristic~\cite[Conjecture 5.10.2]{Cohen2013} tell us
that the class number should be small. Thus, the number of reduction
cycles should also be small. There are many reduced forms when
$\Delta$ is large and there is a global symmetry between reduced form
starting with positive and negative values of $a$. Thus, it is natural
to expect that each reduction cycle should be balanced with respect to
sign. As a consequence, for large values of $\Delta$, there is a good
chance that each desired sign constraint can be
satisfied. Furthermore, if most of them are then $\sum_{i=2}^{n}N_i$
will not deviate much from its lower bound $\sigma(\sigma-1)/2.$\\

Furthermore, it is well known that in any cycle of reduced indefinite forms of
discriminant $\Delta$, there exits a form $(a,b,c)$ with
$|a|\leq \sqrt{\Delta/5}$, see Exercice~17 in~\cite[Chapter
5]{Cohen2013}. As a consequence, in indefinite blocks, it is
reasonable to expect that $|a|$ bounded away from $|c|$. This is also
heuristic, since Exercice~17 does not tell us where such a good form
lies in the reduction cycle and thus does not guarantee us progress in
the limited number of reduction steps the algorithm performs for each
block. However, if even a constant fraction of indefinite blocks correspond
to a form $(a,b,c)$ where $|a|$ is noticeably smaller than
$\sqrt{\Delta/4}$, the overall quality we get outperforms the result of
Theorem~\ref{th:IndefLLL}. \\

Another way to justify Heuristic~\ref{th:heurqual} is to perform
extensive experiments. We have not done so yet, however, the early
experiments we report in the next section tend to indicate that
Heuristic~\ref{th:heurqual}  is in fact too weak. Indeed, in our
examples, $|\bb(\vec{v}_1, \vec{v}_1)|$ is much smaller that the $d$-th root of
$|\det(G_L)|$ itself.

\subsection{Implementation and experiments}
In order to experiment with our algorithm, we have written a first
version of it in interpreted Magma. This implementation is way less
sophisticated than state of the art implementations of LLL. Indeed,
with LLL, many techniques have been used to optimize performance. One
of the most important is probably the use of floating point arithmetic
during the GSO computations. Concerning GSO, advanced implementation
also update Gram-Schmidt orthogonalized vector when a lattice
modification occurs. This is very important for concrete runtimes.

In our current version, as given in Appendix~\ref{sec:magmacode}, we
recompute generalized GSO vectors after each 
lattice modification rather than updating them and all computations
are done using rationals.
In terms of speed performance, this is sub-optimal, however, these
choices made implementation much easier and more timely. Thanks to it,
we can report preliminary results and give comparison to the previous
state of the art. Using Magma is also very useful for the comparison
since this software already contains an implementation of LLL for
indefinite lattices. According to Magma's documentation, this
implementation is a variant of Simon's algorithm with some
undocumented changes to improve the treatment of isotropic
vectors. It is possible to either activate or deactivate this special
treatment by using a parameter in the function call.

\subsubsection{Random Gram matrices} 
Our first test is done by picking a $10\times 10$ symmetric matrix
with entries in $[-100,100]$. For example, the matrix:
$$
\begin{pmatrix}
  44& 25& 45& 93& -71& -49& -69& -99& 3& -3\\
  25& -78& 47& 54& 99& 87& -52& 49& -66& -34\\
  45& 47& 78& 16& 50& 5& -29& 35& -37& 27\\
  93& 54& 16& -59& -34& 34& -81& -68& -43& 76\\
  -71& 99& 50& -34& 68& 80& 29& -77& 22& -90\\
  -49& 87& 5& 34& 80& 30& 53& 90& 20& -80\\
  -69& -52& -29& -81& 29& 53& 28& -61& 75& 80\\
  -99& 49& 35& -68& -77& 90& -61& -49& -66& 9\\
  3& -66& -37& -43& 22& 20& 75& -66& -56& 78\\
-3& -34& 27& 76& -90& -80& 80& 9& 78& 74 
\end{pmatrix}.
$$

On this lattice, Simon's standard implementation in Magma, with or
without the special treatment returns:
$$
\begin{pmatrix}
  -9& 1& 1& -2& 0& 4& 2& 4& -3& -2\\
  1& -13& 2& 0& 1& 5& -2& 3& 2& 6\\
  1& 2& 43& -11& -19& 16& 9& 12& -15& 11\\
  -2& 0& -11& 57& 25& 2& 9& 15& -3& 7\\
  0& 1& -19&  25& 60& -18& 0& 14& -17& 20\\
  4& 5& 16& 2& -18& 245& 106& -100& -44& -12\\
  2& -2&  9& 9& 0& 106& -290& -84& 142& 28\\
  4& 3& 12& 15& 14& -100& -84& -348& -115& -80\\
  -3& 2& -15& -3& -17& -44& 142& -115& 1039& 398\\
  -2& 6& 11& 7& 20& -12& 28& -80 &398& -3314
\end{pmatrix}.
$$

With sign alternance off, our code produces the result\footnote{The unimodular transformation has
  large coefficient and cannot be printed here. }:
$$
\begin{pmatrix}
  1& 0& 0& 0& 0& 0& 0& 0& 0& 0\\
  0& 1& 0& 0& 0& 0& 0& 0& 0& 0\\
  0& 0& -1& 0& 0& 0& 0& 0& 0& 0\\
  0& 0& 0& -1& 0& 0& 0& 0& 0& 0\\
  0& 0& 0& 0& -1& 26& 0& 0& 0& 0\\
  0& 0& 0& 0& 26& 49& 128& -116& -48& 42\\
  0& 0& 0& 0& 0& 128& 1424& 2804& -400& -396\\
  0& 0& 0& 0& 0& -116& 2804& -4644& -3379& 1183\\
  0& 0& 0& 0& 0& -48& -400& -3379& -11066& 382066\\
  0& 0& 0& 0& 0& 42& -396& 1183& 382066& 162445 
\end{pmatrix}.
$$

And with sign alternance on, the result is:
$$
\begin{pmatrix}
  1& 0& 0& 0& 0& 0& 0& 0& 0& 0\\
  0& -1& 0& 0& 0& 0& 0& 0& 0& 0\\
  0& 0& 1& 0& 0& 0& 0& 0& 0& 0\\
  0& 0& 0& -1& 0& 0& 0& 0& 0& 0\\
  0& 0& 0& 0& 1& 0& 0& 0& 0& 0\\
  0& 0& 0& 0& 0& 0& 1& 0& 0& 0\\
  0& 0& 0& 0& 0& 1& 0& 0& 0& 0\\
  0& 0& 0& 0& 0& 0& 0& -334202& 2979470& -53199\\
  0& 0& 0& 0& 0& 0& 0& 2979470& 5071980& 50719326\\
  0& 0&0& 0& 0& 0& 0& -53199& 50719326& -66710417 
\end{pmatrix}.
$$

Even with such a relatively small dimension, the difference is
clear. Simon's algorithm returns a first vector of ``squared-norm'' with
absolute value $9$. By contrast, with sign alternance off, we already
get five vectors with ``squared-norm'' $\pm 1$. And, with sign alternance
on, counting the hyperbolic plane, we improve to seven vectors of
``squared-norm'' $\pm 1$. Note that the tenth-root of the absolute
value of the determinant here is close to $131$.

\subsubsection{Worst case matrices}
We now turn to the worst case example of Section~\ref{sec:heuri}.
Assembling two matrices of dimension 5, we have the following Gram:
$$
\begin{pmatrix}
  32768& 16384& 0& 0& 0& 0& 0& 0& 0& 0\\
  16384& 32768& 12288& 0& 0& 0& 0& 0& 0& 0\\
  0& 12288& 24576& 9216& 0& 0& 0& 0& 0& 0\\
  0& 0& 9216& 18432& 6912& 0& 0& 0& 0& 0\\
  0& 0& 0& 6912& 13824& 0& 0& 0& 0& 0\\
  0& 0& 0& 0& 0& -10368& -5184& 0& 0& 0\\
  0& 0& 0& 0& 0& -5184& -10368& -3888& 0& 0\\
  0& 0& 0& 0& 0& 0& -3888& -7776& -2916& 0\\
  0& 0& 0& 0& 0& 0& 0& -2916& -5832& -2187\\
  0&0& 0& 0& 0& 0& 0& 0& -2187& -4374 
\end{pmatrix}.
$$

As predicted, neither Magma's native implementation nor our code
modify this matrix. However, this is not the end of the story. To
further test Heuristic~\ref{th:heurqual},  we pick
a small random unimodular matrix, use to modify the example.
After that, Magma's implementation outputs:
$$
\begin{pmatrix}
  -4374& 2187& -2187& 2187& 2187& 0& 0& 0& 0& 0\\
  2187& -5832& 2916& -2916& -2916& 0& 0& 0& 0& 0\\
  -2187& 2916& -7776& 3888& 3888& 0& 0& 0& 0& 0\\
  2187& -2916& 3888& -10368& -5184& 0& 0& 0& 0& 0\\
  2187& -2916& 3888& -5184& -10368& 0& 0& 0& 0& 0\\
  0& 0& 0& 0& 0& 13824& -6912& -6912& 6912& -6912\\
  0& 0& 0& 0& 0& -6912& 18432& 9216& -9216& -2304\\
  0& 0& 0& 0& 0& -6912& 9216& 24576& -12288& -5376\\
  0& 0& 0& 0& 0& 6912& -9216& -12288& 32768& 9472\\
  0& 0& 0& 0& 0& -6912& -2304& -5376& 9472& 32768 
\end{pmatrix}.
$$

By constract, without the sign strategy,  our code gives:
$$
\begin{pmatrix}
  2& 5& 0& -1& 0& 0& 0& 0& 0& 0\\
  5& -4& -6& -7& 0& -3& 0& -15& 12& -3\\
  0& -6& -264& -198& 0& 156& 0& 204& -48& 156\\
  -1& -7& -198& -472& 0& 117& 0& 153& -36& 117\\
  0& 0& 0& 0& -10368& 0& 0& 0& 0& 0\\
  0& -3& 156& 117& 0& -11154& 124416& 8166& -5496& 2670\\
  0& 0& 0& 0& 0& 124416& 0&0& 435456& -311040\\
  0& -15& 204& 153& 0& 8166& 0& 2388606& -866712& 277734\\
  0& 12& -48& -36& 0& -5496& 435456& -866712& 2485536& -371832\\
  0& -3& 156& 117& 0& 2670& -311040& 277734& -371832& 6624366 
\end{pmatrix}.
$$

And finally, with the sign strategy, we have:
$$
\begin{pmatrix}
  2& 5& 0& 0& 0& 0& 0& -1& 0& -1\\
  5& -4& 69& 15& 6& -9& 0& -10& 0& -16\\
  0& 69& 78& 186& -156& 90& 288& -165& -144& -153\\
  0& 15& 186& -354& 924& 414& 288& -543& -144& -315\\
  0& 6& -156& 924& 1464& 6156& -288& 906& 144& 450\\
  0& -9& 90& 414& 6156& -12834& 45216& 24777& 18864& 13293\\
  0& 0& 288& 288& -288& 45216& 11520& 322272& -67968& -84960\\
  -1& -10& -165& -543& 906& 24777& 322272& -320704& 2524176&
  -1163146\\
  0& 0& -144& -144& 144& 18864& -67968& 
  2524176& 1195200& 11063664\\
  -1& -16& -153& -315& 450& 13293& -84960& -1163146& 11063664& 
-13101340 
\end{pmatrix}.
$$
Thus, our shortest vector of ``squared-norm'' 2, is much smaller than 
the tenth root of the determinant, that is equal to $10368$. Whereas
the output of Simon's algorithm is of the predicted magnitude.

\subsubsection{Large signature matrices}
In order to better focus on the difference between the Lov{\'a}sz
condition with absolute value and the reduction of indefinite forms,
we construct a dimension 10 matrix with 9 positive eigenvalues and a
single negative. Thus the signature is 8, the maximum for indefinite
matrices of that dimension.
Our input matrix is :
$$
\begin{pmatrix}
  369& 146& -139& 35& -35& 69& -62& 16& 83& 0\\
  146& 327& -32& -66& -59& -48& -120& 83& 26&  0\\
  -139& -32& 376& 51& -33& -163& 67& -42& 37& 0\\
  35& -66& 51& 312& -124& -52& -157& -162&  20& 0\\
  -35& -59& -33& -124& 528& 1& 98& 40& 80& 0\\
  69& -48& -163& -52& 1& 221& -56& 109& -56& 0\\
  -62& -120& 67& -157& 98& -56& 365& -9& 90& 0\\
  16& 83& -42& -162& 40& 109& -9& 188& -42& 0\\
  83& 26& 37& 20& 80& -56& 90& -42& 121& 0\\
  0& 0& 0& 0& 0& 0& 0& 0& 0& -7396 
\end{pmatrix}.
$$

Magma's implementation of Simon does very little and outputs:
$$
\begin{pmatrix}
  121& -3& -41& -15& -42& 53& -22& -19& -26& 0\\
  -3& 220& 42& -46& 5& -104& 7& 58& 0& 0\\
  -41&42& 188& -67& 63& -63& 80& 7& -56& 0\\
  -15& -46& -67& 167& 46& 6& -23& 51& 13& 0\\
  -42& 5& 63& 46& 188& -59& 26& 67& 36& 0\\
  53& -104& -63& 6& -59& 194& -32& -96& 7& 0\\
  -22& 7& 80& -23& 26& -32& 176& 66& -39& 0\\
  -19& 58& 7& 51& 67& -96& 66& 271& -64& 0\\
  -26& 0& -56& 13& 36& 7& -39& -64& 296& 0\\
  0& 0& 0& 0& 0& 0& 0& 0& 0& -7396
\end{pmatrix}.
$$

With our code, irrespective of the chosen sign strategy, we obtain:
$$
\begin{pmatrix}
  -4& 30& -2& -3& 0& 2& 3& 3& 1& 0\\
  30& 17& 120& 135& -61& -58& -253& -123& -177& -231\\
  -2& 120& 240& 242& -104& 159& -206& 33& -93& 52\\
  -3& 135& 242& 747& 299& -34& -355& 177& -180& -171\\
  0& -61& -104& 299& 597& -383& 51& 140& 120& -253\\
  2& -58& 159& -34& -383& 539& -181& 72& 32& 222\\
  3& -253& -206& -355& 51& -181& 559& -54& 139& 243\\
  3& -123& 33& 177& 140& 72& -54& 292& 78& -43\\
  1& -177& -93& -180& 120& 32& 139& 78& 487& -28\\
  0& -231& 52& -171& -253& 222& 243& -43& -28& 22041 
\end{pmatrix}.
$$
In particular, we find a short vector of ``squared-norm'' $-4$ that
Simon's algorithm missed. Interestingly, the tenth-root of the
determinant is around $224$, much larger that our shortest vector.

\section{Conclusion}
In this paper, we revisited indefinite lattice reduction and came up
with the somewhat surprinsing conclusion that it seems to be much
easier to find short vectors in this situation than in the usual case of
definite lattices. Even our strong
Heuristic~\ref{th:heurqual} does not seem to reflect the full story.

As a tentative explanation, we would like to exhibit one big
difference between the definite and indefinite case. In the definite
case, there are only finitely many good bases (for whatever notion we
want to choose). This number might be very large but it remains
finite and decreases as we put stronger restrictions on what a good
basis should be.

In the indefinite case, this is no longer true. To see that, consider
the following Gram matrix:
$$
G=\begin{pmatrix}
  1 & 0& 0 & 0\\
  0 &-1 &0 & 0\\
  0 & 0& 1&  0\\
  0 & 0 & 0 &-1
\end{pmatrix}.
$$

It is clearly a very good basis. Now form the unimodular matrix:
$$
U=\begin{pmatrix}
  1&  0&  1&  1\\
  -1&  1&  0& -1\\
  1& -1& -1&  0\\
  0 & 1 & 1 & 1
\end{pmatrix}
$$
We easily check that $\transp{U}\cdot G\cdot U=G$. Furthermore,
starting from the characteristic polynomial of $U$, we can derive the identity:
$$
(U^k)_{1,1}=((-1)^k+7+2\,k^2)/8.
$$
This shows that $U$ has infinite multiplicate order. Thus, there exists
infinitely many unimodular matrices that map $G$ to itself.
In other words, we have found infinitely many good bases of this
indefinite lattice.

Furthermore these good bases are probably all over the place. So, it makes some
sense to expect that any basis for the lattice is not too far from
some good basis. If such is the case, it is natural to assume that
lattice reduction becomes easier.

\subsection*{Future directions of research} Since its invention, LLL
has been the object of many studies and improvements. We can only hope
that future research on indefinite lattice reduction will be as
active.

On the application front, having a better algorithm opens new
hope. For example, the number theoretic applications proposed
in~\cite{IvanyosSz96,Simon05} can probably be improved upon by the new
ability to find shorter vectors. Lattice
reduction also has many applications in complexity and cryptography
and it would be nice to find similar uses of indefinite lattice
reduction.

On the algorithmic front, the study of LLL has been very rich. Faster
algorithms based on approximate floating point computation, refined
complexity analyses, stronger algorithms with better approximation
factors. Lattice reduction has also been extended to solve other
problems such as the closest vector problem. Can similar
improvements and extensions also be applied to indefinite lattices?

\bibliographystyle{amsalpha}
\bibliography{indef}

\newcommand{\etalchar}[1]{$^{#1}$}
\providecommand{\bysame}{\leavevmode\hbox to3em{\hrulefill}\thinspace}
\providecommand{\MR}{\relax\ifhmode\unskip\space\fi MR }
\providecommand{\MRhref}[2]{%
  \href{http://www.ams.org/mathscinet-getitem?mr=#1}{#2}
}
\providecommand{\href}[2]{#2}
\begin{thebibliography}{DPTZ22}

\bibitem[Coh13]{Cohen2013}
Henri Cohen, \emph{A course in computational algebraic number theory}, vol.
  138, Springer, 2013.

\bibitem[Dav78]{davis1978rational}
Clive~S Davis, \emph{Rational approximations to e}, Journal of the Australian
  Mathematical Society \textbf{25} (1978), no.~4, 497--502.

\bibitem[DPTZ22]{denton2022eigenvectors}
Peter~B. Denton, Stephen~J. PArke, Terence Tao, and Xining Zhang,
  \emph{Eigenvectors from eigenvalues: A survey of a basic identity in linear
  algebra}, AMS \textbf{59} (2022), no.~1, 31--58.

\bibitem[EJ20]{espitau2020certified}
Thomas Espitau and Antoine Joux, \emph{Certified lattice reduction}, Advances
  in Mathematics of Communications \textbf{14} (2020), no.~1, 137--159.

\bibitem[Gau01]{Gauss1801}
Carl~Friedrich Gauss, \emph{Disquisitiones arithmeticae}, Lipsiae, 1801.

\bibitem[Han09]{hanrot2009lll}
Guillaume Hanrot, \emph{{LLL}: a tool for effective diophantine approximation},
  The {LLL} Algorithm: Survey and Applications, Springer, 2009, pp.~215--263.

\bibitem[IS96]{IvanyosSz96}
G\'abor Ivanyos and \'Agnes Sz\'ant\'o, \emph{Lattice basis reduction for
  indefinite forms and an application}, Discrete Mathematics \textbf{153}
  (1996), no.~1, 177--188.

\bibitem[KZ77]{KorkineZ1877}
A.~Korkine and G.~Zolotareff, \emph{Sur les formes quadratiques positives},
  Mathematische Annalen \textbf{11} (1877), no.~2, 242--292.

\bibitem[LLL82]{lenstra1982factoring}
Arjen~K. Lenstra, Hendrik~W. Lenstra, and L{\'a}szl{\'o} Lov{\'a}sz,
  \emph{Factoring polynomials with rational coefficients}, Math. ann
  \textbf{261} (1982), no.~4, 515--534.

\bibitem[Mar79]{Markoff1879}
Andrey Markoff, \emph{Sur les formes quadratiques binaires ind{\'e}finies},
  Mathematische Annalen \textbf{15} (1879), no.~3, 381--406.

\bibitem[Ngu09]{nguyen2009hermite}
Phong~Q. Nguyen, \emph{Hermite’s constant and lattice algorithms}, The {LLL}
  Algorithm: Survey and Applications, Springer, 2009, pp.~19--69.

\bibitem[Sim05]{Simon05}
Denis Simon, \emph{Solving quadratic equations using reduced unimodular
  quadratic forms}, Math. Comput. \textbf{74} (2005), no.~251, 1531--1543.

\bibitem[SLL{\etalchar{+}}09]{smeets2009history}
Ionica Smeets, Arjen Lenstra, Hendrik Lenstra, L{\'a}szl{\'o} Lov{\'a}sz, and
  Peter van Emde~Boas, \emph{The history of the {LLL}-algorithm}, The {LLL}
  Algorithm: Survey and Applications, Springer, 2009, pp.~1--17.

\end{thebibliography}

\vspace{1cm}
\appendix
\section{Signature under unimodular transformations and GSO}
\begin{theorem} \label{th:unimodEigen}
If $G$ is an invertible real symmetric matrix and $U$ is a unimodular
matrix (of the same dimension as $G$), then $\transp{U}\cdot G \cdot
U$ has the same number of positive and negative eigenvalues as $G$. As
an immediate consequence, $G$  and $\transp{U}\cdot G \cdot U$ have
the same signature.
\end{theorem}
\begin{proof}
  We first recall that unimodular matrices are generated by permutations and
  elementary translations, i.e. matrices of the form:
  $$U_\tau=
  \begin{pmatrix}
    1& 0 & 0 &\cdots &0\\
    \tau& 1 & 0 &\cdots &0\\
    0 & 0 &1 &\cdots &0\\
    \vdots &\vdots &\vdots &\ddots &\vdots\\
    0 & 0 &0 &\cdots &1
  \end{pmatrix}.
  $$
  This elementary translation has a single non-zero
  off-diagonal entry, and it correspond to adding $\tau$ times the
  second basis vector to the first one.

  Thus, it suffices to prove the theorem when $U$ is either a
  permutation or an elementary translation. Of course, applying a
  permutation to $G$ does not change its eigenvalues,thus  we only need to
  consider the case where $U$ is an elementary translation $U_\tau$.

  Decompose $G$ (of dimension $n$) into blocks in the following way:
  $$
  G=  \begin{pmatrix}
    g_{11} & v\\
    \transp{v}& G'
  \end{pmatrix},
  $$
  where $g_{11}$ is the top-left entry of $G$, $v$ the rest of the
  first row and $G'$ the $(n-1)\times (n-1)$ minor obtained by
  removing the first row and column of $G$.

  An easy computation shows that $\transp{U_\tau }\cdot G\cdot U_\tau$
  is of the form:
    $$
  \begin{pmatrix}
    g'_{11} & v'\\
    \transp{v'}& G'
  \end{pmatrix},
  $$
Most importantly, $G$ and $\transp{U_\tau }\cdot G\cdot U_\tau$ share
the minor $G'$.

Let $f_1\leq f_2\leq \cdot \leq f_{n-1}$ denote the eigenvalues of
$G'$, $e_1\leq e_2\leq \cdot \leq e_{n}$  denote the eigenvalues of
$G$ and $e'_1\leq e'_2\leq \cdot \leq e'_{n}$ the eigenvalues of
$\transp{U_\tau }\cdot G\cdot U_\tau$.

Thanks to Cauchy's interlacing theorem, e.g. see~\cite{denton2022eigenvectors}, we know that:
\begin{align*}
& e_1\leq f_1\leq e_2\leq \cdots\leq e_{n-1}\leq f_{n-1}\leq e_n\quad 
                 \mbox{and} \\
& e'_1\leq f_1\leq e'_2\leq \cdots\leq e'_{n-1}\leq f_{n-1}\leq e'_n.
\end{align*}

We now distinguish two cases. First, if there is an index $i$ such
that $f_i=0$ then we have $e_1<0$, \dots, $e_i<0$ and $e_{i+1}>0$,
\dots, $e_n>0$ (and likewise for $e'$). In that case, $e$ and $e'$
contains $i$ negative eigenvalues and $n-i$ positive ones. Thus, the
conclusion follows.

In the second case, we let $i$ denote the first positive $f_i$, it covers the three following subcases:
\begin{itemize}
\item If $f_1>0$, we set $i=1.$
\item if $f_{i-1}<0<f_{i},$ we directly have $i.$
\item If $f_{n-1}<0$, we set $i=n$.
\end{itemize}
Then, we see that $e_j<0$ when $j\leq i-1$ and $e_j>0$ when $j\geq
i+1$ and similarly for $e'$. However, the signs of $e_i$ and $e'_i$
are not directly known. Still, we know that:
$$
\prod_{j=1}^n e_j=\det{G}=\det{\left(\transp{U_\tau }\cdot G\cdot U_\tau\right)}=\prod_{j=1}^n e'_j.
$$
As a consequence, $e_i$ and $e'_i$ also have the same sign. This concludes
the proof in the second case.
\end{proof}

\begin{theorem}  \label{th:GSOEigen}
Let $G$ is an invertible real symmetric matrix with generalized GSO of
the form:
$$
G^{*}=\transp{T}\cdot G\cdot T,
$$
where $T$ is lower triangular with $1$s on the diagonal and $G^{*}$ is
a diagonal join of either pure diagonal entries $g^{*}_{ii}$ or small
$2\times 2$ hyperbolic planes of the form
$\begin{pmatrix}
  0 & g^{*}_{i,i+1}\\
  g^{*}_{i,i+1} & 0\\
\end{pmatrix}.
$

Denote by $n^{+}$ and $n^{-}$ the number of positive and negative
eigenvalues of $G$, by $n^{*}_{\mbox{pos}}$, $n^{*}_{\mbox{neg}}$,
the number of positive and negative diagonal elements in $G^{*}$ and
by $n^{*}_{\mbox{hyp}}$ the number of hyperbolic planes in $G^{*}$.
Then:
\begin{align*}
  &n^{+}= n^{*}_{\mbox{pos}}+ n^{*}_{\mbox{hyp}}\quad\mbox{and}\\
  &n^{-}= n^{*}_{\mbox{neg}}+ n^{*}_{\mbox{hyp}}.
\end{align*}
\end{theorem}
\begin{proof}
  The proof is in two steps. First, note that $T$ can be decomposed
  into a product of matrices of the form $U_{\tau}$ as in the previous
  proof, except that $\tau$ is not necessarily an integer here. Thus,
  $G$ and $G^{*}$ have the same number of positive and negative
  eigenvalues.

  Second, we can remark that each hyperbolic plane contributes one
  positive and one negative eigenvalue. This concludes the proof.
\end{proof}

\section{Post-processing to remove some hyperbolic planes}
Once a reduced basis is reached, the presence of hyperbolic planes can
be an inconvenience. We give here useful transformations that can
remove them with the help of an extra vector outside of the plane.

To specify these transformation, we just need to consider
three-dimensional Gram matrices. We call the corresponding vectors
$\vec{u},$ $\vec{v},$ and $\vec{w}$ in the  rest of this Appendix.
Up to scaling, the Gram matrix must be of the form:
$$
G_3=\begin{pmatrix}
  1 &0 &0 \\
  0 & 0& \alpha\\
  0 &\alpha&0
\end{pmatrix}.
$$
We are especially interested with the case where $|\alpha$ is close to
$1$. Remark that we can ensure that $\alpha>0$ by replacing $\vec{w}$
by $-\vec{w}$ if needed. We let $\epsilon_{\alpha}=1-\alpha$.

Now, consider the unimodular transformation:
$$
U=\begin{pmatrix}
  1 & 1& 1\\
  1 & 1 & 0\\
  -1&0 & -1
\end{pmatrix}.
$$

It corresponds to the basis $(\vec{u}+\vec{v}-\vec{w},\vec{u}+\vec{v}, \vec{u}-\vec{w}).$ The Gram matrix after the transformation become:
$$G_U=
\begin{pmatrix}
  1-2\alpha & 1-\alpha& 1-\alpha\\
  1-\alpha & 1 & 1-\alpha\\
  1-\alpha& 1-\alpha & 1
\end{pmatrix}=
\begin{pmatrix}
  -1+2 \epsilon_{\alpha} & \epsilon_{\alpha}& \epsilon_{\alpha}\\
  \epsilon_{\alpha}& 1 & \epsilon_{\alpha}\\
  \epsilon_{\alpha}& \epsilon_{\alpha} & 1
\end{pmatrix}.
$$

In particular, when $\alpha=1$,  we find the diagonal matrix:
$$G_U=
\begin{pmatrix}
  -1& 0 & 0\\
 0&1&0\\
  0&0&1
\end{pmatrix}.
$$

\section{Magma code}\label{sec:magmacode}

\begin{lstlisting}[language=Delphi,morekeywords={cat,elif,le,lt,ge,gt,ne,true,false,error,return}]
QQ:=Rationals();
ZZ:=Integers();
LLLCste:=99/100;

procedure CheckSymmetric(G0)
  dim:=NumberOfRows(G0);
  for i:=1 to dim-1 do
    for j:=i+1 to dim do
      error if (G0[i,j] ne G0[j,i]), "Gram matrix should be symmetric";
    end for;
  end for;
end procedure;

function QuadFormIsReduced(A,B,C)
  Det:=B*B-4*A*C;
  reduced:=false;
  if (A eq 0) and (B eq 0) then
    reduced:=true;
  elif (Det lt 0) and (Abs(B) le Abs(A)) and (Abs(A) le Abs(C)) then
    reduced:=true;
  elif (Det gt 0) then
    issquare,Squ:=IsSquare(Det);
    if not issquare then
      if (B gt 0) and (B*B lt Det) and
	 ((2*Abs(A)-B)^2 lt Det) and ((2*Abs(A)+B)^2 gt Det) then
	reduced:=true;
      end if;
    else
      if (B eq 0) and (2*Abs(A) eq Squ) and (A eq -C) then
	reduced:=true;
      elif (B eq Squ) and (Abs(A) lt Squ/2) and (C eq 0) then
	reduced:=true;
      elif ((B gt 0) and (B lt Squ) and
	    (2*Abs(A) gt Squ-B) and (2*Abs(A) lt Squ+B)) then
	reduced:=true;
      end if;
    end if;
  end if;
  return reduced;
end function;


function QuadIsReduced(Norm1,Norm2,Scal)
  return QuadFormIsReduced(Norm1, 2*Scal, Norm2);
end function;


function QuadraticFormReductionStep(A,B,C)
  // Already definite reduced (or zero degenerate)
  Det:=B*B-4*A*C;
  needsignchange:=false;
  if (Det le 0) and QuadFormIsReduced(A,B,C) then
    return true,ScalarMatrix(2,1),A,B,C;
  end if;
  if Det le 0 then
    // Definite case reduction extended to the zero degenerate case
    // Note: A cannot be zero in degen case or we would already be reduced
    lambda:=Floor(-B/(2*A)+(1/2));
    AA,BB,CC:=A,B+2*lambda*A,C+lambda*B+lambda^2*A;
    if Abs(AA) gt Abs(CC) then
      return QuadFormIsReduced(CC,-BB,AA),
	     Matrix(ZZ,2,2,[lambda,-1,1,0]),
	     CC,-BB,AA;
    else
      return QuadFormIsReduced(AA,BB,CC),
	     Matrix(ZZ,2,2,[1,lambda,0,1]),
	     AA,BB,CC;
    end if;
  else
    // If indefinite go through reduction cycle
    issquare,Squ:=IsSquare(Det);
    if not issquare then 
      Squ:=Sqrt(Det);
      if Abs(C) gt Squ then
      	lambda:=Round(B/(2*C));
      else
	lambda:=(Squ+B)/(2*C);
	if C gt 0 then
          lambda:=Floor(lambda);
	else
          lambda:=Ceiling(lambda);
	end if;
      end if;
      AA,BB,CC:=C,-B+2*lambda*C,A-lambda*B+lambda^2*C;
      return QuadFormIsReduced(AA,BB,CC),
	     Matrix(ZZ,2,2,[0,1,-1,-lambda]),
	     AA,BB,CC;
    else //Exact Squ situation
      if (B eq 0) and (A+C eq 0) then
        return true, Matrix(ZZ,2,2,[0,1,-1,0]),
	       C,-B,A;
      elif (B eq -Squ) then // IMPROPER EQUIVALENCE HERE
        return QuadFormIsReduced(A,-B,C),
	       Matrix(ZZ,2,2,[1,0,0,-1]),
	       A,-B,C;
      elif (C eq 0) then 
	if (Abs(A) ne Squ/2) then // B is +Squ, No Exchange Here
	  lambda:=Round(A/B);
	  AA,BB,CC:=A-lambda*B,B,C;
          return QuadFormIsReduced(AA,BB,CC),
		 Matrix(ZZ,2,2,[1,0,-lambda,1]),
		 AA,BB,CC;
	else
	  lambda:=Round(B/(2*A)); 
	  AA,BB,CC:=C-lambda*B+lambda^2*A,-B+2*lambda*A,A;
	  return QuadFormIsReduced(AA,BB,CC),
		 Matrix(ZZ,2,2,[-lambda,-1,1,0]),
		 AA,BB,CC;
	end if;
      elif (A eq 0) then // B is +Squ, IMPROPER EQUIVALENCE HERE
	return QuadFormIsReduced(C,B,A),
	       Matrix(ZZ,2,2,[0,1,1,0]),
	       C,B,A;
      else //Same as non square
	if Abs(C) gt Squ then
      	  lambda:=Round(B/(2*C));
	else
          lambda:=(Squ+B)/(2*C);
	  if C gt 0 then
            lambda:=Floor(lambda);
	  else
            lambda:=Ceiling(lambda);
	  end if;
	end if;
	AA,BB,CC:=C,-B+2*lambda*C,A-lambda*B+lambda^2*C;
	return QuadFormIsReduced(AA,BB,CC),
	       Matrix(ZZ,2,2,[0,1,-1,-lambda]),
	       AA,BB,CC;
      end if;
    end if;
  end if;
end function;

function QuadraticReductionStep(Norm1,Norm2,Scal)
  bool,U,A,B,C:=QuadraticFormReductionStep(Norm1,2*Scal,Norm2);
  return bool,U,A,C,B/2;
end function;

function ExtendedGSO(L, Lstar, HyperbolicMarkers, k, pos)
  newGSOvec:=L[pos];
  redpos:=1;
  while (redpos le k) do
    if HyperbolicMarkers[redpos] then //hyperbolic plane extended GSO
      crossprod:=(Lstar[redpos],Lstar[redpos+1]);
      scal1:=(newGSOvec,Lstar[redpos]);
      scal2:=(newGSOvec,Lstar[redpos+1]);
      newGSOvec-:=(scal1/crossprod)*Lstar[redpos+1]+
		  (scal2/crossprod)*Lstar[redpos];
      redpos+:=2;
    else //Regular GSO
      norm:=(Lstar[redpos],Lstar[redpos]);
      scal:=(newGSOvec,Lstar[redpos]);
      newGSOvec-:=(scal/norm)*Lstar[redpos];
      redpos+:=1;
    end if;
  end while;
  return newGSOvec;
end function;

procedure SizeReduce(L, Lstar, HyperbolicMarkers, k, ~newREDvec, ~kzero)
  redpos:=k;
  while (redpos gt 0) do
    if (redpos gt 1) and HyperbolicMarkers[redpos-1] then
      //hyperbolic plane red
      crossprod:=(Lstar[redpos],Lstar[redpos-1]);
      scal1:=(newREDvec,Lstar[redpos]);
      scal2:=(newREDvec,Lstar[redpos-1]);
      newREDvec-:=Truncate(scal1/crossprod)*L[redpos-1]+
		  Truncate(scal2/crossprod)*L[redpos];
      if ((scal1/crossprod) ne Truncate(scal1/crossprod)) or
	 ((scal2/crossprod) ne Truncate(scal2/crossprod)) then
	kzero:=false; end if;
      redpos-:=2;
    else //Regular size red
      norm:=(Lstar[redpos],Lstar[redpos]);
      scal:=(newREDvec,Lstar[redpos]);
      newREDvec-:=Truncate(scal/norm)*L[redpos];
      if ((scal/norm) ne Truncate(scal/norm)) then kzero:=false; end if;
      redpos-:=1;
    end if;
  end while;
  if (newREDvec,newREDvec) ne 0 then kzero:=false; end if;
end procedure;

procedure CleanUpReduction(L, Lstar, HyperbolicMarkers, NewGSOvec,
			   k, ~newREDvec, ~kzero)
  Norm1:=(Lstar[k],Lstar[k]);
  Scal:=(Lstar[k],newREDvec);
  Norm2:=(NewGSOvec,NewGSOvec)+Scal*Scal/Norm1;

  Det:=Scal*Scal-Norm1*Norm2;
  if Det le 0 then
    lambda:=-Truncate(Scal/Norm1);
    if (Det ne 0) or ((Scal/Norm1) ne Truncate(Scal/Norm1)) then
      kzero:=false; end if;
  else
    issquare,Squ:=IsSquare(Det);
    if not issquare then Squ:=Sqrt(Det); end if;
    if (Scal eq 0) and (Norm1+Norm2) eq 0 then
      lambda:=0;
    elif (Norm2 ne 0) or (Abs(Scal) ne Squ) or (Abs(Norm1) ne Squ) then
      if Abs(Norm1) gt Squ then
	lambda:=Round(-Scal/Norm1);
      else
	lambda:=(Squ-Scal)/Norm1;
	if Norm1 gt 0 then
	  lambda:=Floor(lambda);
	else
	  lambda:=Ceiling(lambda);
	end if;
      end if;
    else
      lambda:=Round(-Scal/Norm1);
    end if;
    kzero:=false;
  end if;
  newREDvec+:=lambda*L[k];
end procedure;

// Assume that position and position+1 are already size reduced
// and GSO of position and position+1 are known
procedure FullReductionStep(~nextmove, ~L, Lstar, HyperbolicMarkers,
			    position, signenforce, countlimit)
  Norm1:=(Lstar[position],Lstar[position]);
  Scal:=(Lstar[position],L[position+1]);
  Norm2:=(Lstar[position+1],Lstar[position+1])+Scal*Scal/Norm1;

  nextmove:=1;
  Det:=-Norm1*Norm2+Scal*Scal;
  alreadyreduced:=QuadIsReduced(Norm1,Norm2,Scal);
  if (Det le 0) and alreadyreduced then return; end if;
  if (Det le 0) then // Standard Lovasz reduction (one step)
    lambda:=Truncate(Scal/Norm1);
    L[position+1]-:=lambda*L[position];
    Norm2:=Norm2-2*lambda*Scal+lambda*lambda*Norm1;
    Scal-:=lambda*Norm1;
    if (Abs(Norm2) le LLLCste*Abs(Norm1)) then
      nextmove:=-1;
      tmp:=L[position];L[position]:=L[position+1];L[position+1]:=tmp;
      tmp:=Norm1;Norm1:=Norm2;Norm2:=tmp;
    end if;
    return;
  else //Indefinite 2x2 reduction
    worthwhile:=false;
    localU:=ScalarMatrix(2,1);
    OldNorm1:=Norm1;
    OldNorm2:=Norm2;
    if (not signenforce) or (Norm1 eq 0) then
      error if (Norm1 eq 0), "Norm1=0 should not occur her";
      localsignenforce:=false;
    else
      prevpos:=position-1;
      while (prevpos gt 1) and HyperbolicMarkers[prevpos-1] do
	prevpos:=prevpos-2;
      end while;
      if prevpos eq 0 then
        localsignenforce:=false;
      else
	localsignenforce:=true;
      	prevNorm:=(Lstar[prevpos],Lstar[prevpos]);
	if (Sign(Norm1) eq Sign(prevNorm)) then OldNorm1:=-OldNorm1; end if;
      end if;
    end if;
    if not alreadyreduced then
      isreduced:=false;
      while not isreduced do
	isreduced,partU,Norm1,Norm2,Scal:=
			QuadraticReductionStep(Norm1,Norm2,Scal);
	localU:=localU*partU;
        if (Abs(Norm1) lt Abs(OldNorm1)/2) then
	  if ((not localsignenforce) or (Norm1 eq 0) or
	      (Sign(Norm1) eq Sign(OldNorm1))) then
	    worthwhile:=true;
	    break;
	  end if;
	end if;
      end while;
      if (not worthwhile) and localsignenforce then
	if not ((Norm1 eq 0) or (Sign(Norm1) eq Sign(OldNorm1))) and 
	   ((Norm2 ne 0) and (Abs(Norm2) lt LLLCste*Abs(OldNorm1)))
	   and (Norm2*Norm2 le Det) then
	  //Bad sign and second vector worthwhile
	  localU*:= Matrix(ZZ,2,2,[0,1,1,0]);
	  tmp:=Norm1;Norm1:=Norm2;Norm2:=tmp;
	  worthwhile:=true;
	end if;
      end if;
      if (not worthwhile) and
	 ((Abs(Norm1) lt LLLCste*Abs(OldNorm1)) or
	  ((Abs(Norm1) le Abs(OldNorm1)) and
	   (((OldNorm2 eq 0) and (Norm2 ne 0))))) then
	worthwhile:=true;
      end if;
    else
      for count:=1 to countlimit do
      	isreduced,partU,Norm1,Norm2,Scal:=
		QuadraticReductionStep(Norm1,Norm2,Scal);
	localU:=localU*partU;
	if localsignenforce then //Do the sign change enforcement
	  if (Norm1 eq 0) or
	     ((Sign(Norm1) eq Sign(OldNorm1)) and
	      (Abs(Norm1) le LLLCste*Abs(OldNorm1)))  then
	    worthwhile:=true;
	  end if;
	elif (Abs(Norm1) le LLLCste*Abs(OldNorm1)) then worthwhile:=true;
	end if;
	if worthwhile then break; end if;
      end for;
    end if;
    if worthwhile then
      newposvec:=localU[1,1]*L[position]+localU[2,1]*L[position+1];
      newposplusvec:=localU[1,2]*L[position]+localU[2,2]*L[position+1];
      iskzero:=true;
      SizeReduce(L, Lstar, HyperbolicMarkers,
		 position-1, ~newposvec, ~iskzero);
      if iskzero and ((newposvec,newposvec) ne 0) then
	iskzero:=false;
      end if;
      if not iskzero then
	if ((L[position] ne newposvec) and (L[position] ne -newposvec))
	   or ((L[position+1] ne newposplusvec) and
	       (L[position+1] ne -newposplusvec)) then
	  nextmove:=-1;
 	  L[position]:=newposvec; L[position+1]:=newposplusvec;
	end if;
      else
	if ((L[position+1] ne newposvec) and
	    (L[position+1] ne -newposvec))
	   or ((L[position] ne newposplusvec) and
	       (L[position] ne -newposplusvec)) then
	  nextmove:=-1;
	  L[position]:=newposplusvec; L[position+1]:=newposvec;
	end if;
      end if;
    end if;
  end if;
end procedure;



function IndefLLL(G0, Lcoord, maxloop, signenforce, repeatlimit)
  dim:=NumberOfRows(G0);
  error if (dim ne NumberOfColumns(G0)), "Gram should be a Square Matrix";
  CheckSymmetric(G0);
  Vsp:=VectorSpace(QQ,dim,G0);
  tmp:=[0:i in [1..dim]];
  if Lcoord ne [] then
    Nvec:=#Lcoord;
    L:=[Vsp!Lelt:Lelt in Lcoord];
  else
    Nvec:=dim;
    L:=[];
    for i:=1 to dim do
      tmp[i]:=1;
      Append(~L,Vsp!tmp);
      tmp[i]:=0;
    end for;
  end if;
  HyperbolicMarkers:=[false:i in [1..Nvec]];
  Lstar:=[];
  position:=1;
  countround:=0;
  maxposition:=1;
  while (position le Nvec) do
    if (position gt maxposition) then
      maxposition:=position;
      "@",[maxposition,countround];
    end if;
    countround+:=1;
    if (maxloop ne 0) and (countround gt maxloop) then break; end if;
    candpos:=position;
    found:=false;
    Hyperbolic:=false;
    while not found do           //discover non-k-zero vector
      if IsZero(L[candpos]) then // Remove exact dependencies
	L:=L[1..candpos-1] cat L[candpos+1..Nvec];
	Nvec:=Nvec-1;
      	if candpos gt Nvec then break; end if;
	continue;
      end if;

      if position eq 1 then
	norm:=(L[candpos],L[candpos]);
	NewGSOvec:=L[candpos];
	found:=norm ne 0;
      else
	iskzero:=true;
	NewGSOvec:= ExtendedGSO(L, Lstar, HyperbolicMarkers,
				position-1, candpos);
	NewVec:=L[candpos];
	if (position gt 2) and HyperbolicMarkers[position-2] then
	  SizeReduce(L, Lstar, HyperbolicMarkers, position-1,
		     ~NewVec, ~iskzero);
	else
	  CleanUpReduction(L, Lstar, HyperbolicMarkers, NewGSOvec,
			   position-1, ~NewVec, ~iskzero);
	  SizeReduce(L, Lstar, HyperbolicMarkers, position-2,
		     ~NewVec, ~iskzero);
	end if;
	L[candpos]:=NewVec;
	if (NewGSOvec,NewGSOvec) ne 0 then iskzero:=false; end if;
	found:=not iskzero;
      end if;
      if not found then
	if (candpos gt position) and ((L[candpos],L[candpos-1]) ne 0) then
	  found:=true;
	  Hyperbolic:=true;
	  foundpos1:=candpos-1; foundpos2:=candpos;
	end if;
      end if;
      if not found then
	candpos+:=1;
	if candpos gt Nvec then break; end if;
      end if;
    end while;
    if not found then
      // Find hyperbolic plane
      for pos1:=position to Nvec-1 do
	for pos2:=pos1+1 to Nvec do
	  if (L[pos1],L[pos2]) ne 0 then
	    foundpos1:=pos1; foundpos2:=pos2;
	    found:=true; Hyperbolic:=true; break;
	  end if;
	end for;
	if found then break; end if;
      end for;
    end if;
    
    if not found then
      break;
    elif Hyperbolic then
      L:=L[1..position-1] cat [L[foundpos1], L[foundpos2]] cat
	 L[position..foundpos1-1] cat
	 L[foundpos1+1..foundpos2-1] cat L[foundpos2+1..Nvec];
      if (position gt 2) and HyperbolicMarkers[position-1] then
	prevscal:=(L[position-2],L[position-1]);
	scal:=(L[position],L[position+1]);
	if Abs(scal) lt Abs(prevscal) then
	  tmp:=L[position]; L[position]:=L[position-2];
	  L[position-2]:=tmp;
	  tmp:=L[position+1]; L[position+1]:=L[position-1];
	  L[position-1]:=tmp;
	  position:=position-2;
	  Prune(~Lstar); Prune(~Lstar);
	  continue;
	end if;
      elif (position gt 1) then
	scal:=(L[position],L[position+1]);
	prevnorm:=(Lstar[position-1],Lstar[position-1]);
	if Abs(scal) lt Abs(prevnorm) then
	  tmp:=L[position+1]; L[position+1]:=L[position-1];
	  L[position-1]:=tmp;
	  position:=position-1;
	  Prune(~Lstar);
	  continue;
	end if;
      end if;
      Append(~Lstar,L[position]);
      Append(~Lstar,L[position+1]);
      HyperbolicMarkers[position]:=true;
      position+:=2;
      continue;
    end if;
    if candpos gt position then
      L:=L[1..position-1] cat [L[candpos]] cat
	 L[position..candpos-1] cat L[candpos+1..Nvec];
    end if;
    Append(~Lstar,NewGSOvec);
    if position eq 1 then
      position+:=1;
    elif (position gt 2) and HyperbolicMarkers[position-2] then
      prevscal:=(L[position-2],L[position-1]);
      scal1:=(L[position-2],L[position]);
      scal2:=(L[position-1],L[position]);
      norm:=(Lstar[position],Lstar[position]);
      active:=false;
      if (scal1 eq 0) and (scal2 ne 0) then
	tmp:=L[position];L[position]:=L[position-2];
	L[position-2]:=tmp;
	active:=true;
      elif (Abs(norm) le Abs(prevscal)) or (scal1 ne 0) then
	tmp:=L[position];L[position]:=L[position-1];
	L[position-1]:=L[position-2];L[position-2]:=tmp;
	active:=true;
      end if;
      if (active) then
      	Prune(~Lstar); Prune(~Lstar); Prune(~Lstar);
	position-:=2;
	HyperbolicMarkers[position]:=false;
      else
	position+:=1;
      end if;
    else
      nextmove:=0;
      HyperbolicMarkers[position-1]:=false;
      FullReductionStep(~nextmove, ~L, Lstar, HyperbolicMarkers,
			position-1, signenforce, repeatlimit);
      if nextmove eq 1 then
	position+:=1;
      elif nextmove eq -1 then
	Prune(~Lstar); Prune(~Lstar);
	position-:=1;
	position:=Max(1,position);
      elif nextmove eq 0 then
	Prune(~Lstar);
      else
	error "Nextmove not defined";
      end if;
    end if;
  end while;
  return L, Lstar, HyperbolicMarkers, countround;
end function;

\end{lstlisting}
\end{document}